\documentclass[journal,10pt]{IEEEtran}
\usepackage[final]{graphicx}
\usepackage{multicol}
\graphicspath{{Images//}} 

\usepackage{hyperref}
\usepackage{psfrag}
\usepackage{epstopdf}
\usepackage{amsfonts}
\usepackage{mathrsfs}
\usepackage{pifont}
\usepackage{verbatim}
\usepackage{upgreek}
\usepackage{color}
\usepackage[usenames,dvipsnames]{xcolor}
\usepackage{algorithm}
\usepackage{algpseudocode}
\usepackage{mdwmath}
\usepackage{mdwtab}
\usepackage{amssymb,amsmath}
\usepackage{amsmath}
\usepackage{amsthm}
\usepackage{epstopdf}
\usepackage{cite}
\usepackage{tikz}
\usepackage{scalefnt}
\usepackage{subfigure}
\usepackage{MnSymbol}
\usepackage{setspace}
\usepackage{mathtools}
\usepackage{psfrag}
\usepackage{todonotes}
\usepackage{soul} 
\usepackage{bbm} %
\usepackage{romannum}%

\usepackage[font=scriptsize]{caption}
\usepackage{algpseudocode}

\UseRawInputEncoding

\algrenewcommand\algorithmicforall{\textbf{foreach}}
\algrenewcommand\algorithmicindent{.8em}

\label{newcommand}

\label{English Chars}

\newcommand{\bH}{\textbf{H}}

\newcommand{\bp}{\textbf{p}}

\newcommand{\bv}{\textbf{v}}

\label{Roman Chars}

\newcommand{\btheta}{\mbox{\boldmath{$\theta$}}}

\label{Theorems}
\newtheorem{theorem}{Theorem}[section]

\newtheorem{lemma}{Lemma}

\newtheorem{mydef}{Definition}

\theoremstyle{remark}

\theoremstyle{remark}
\newtheorem{Remark}{Remark}
\newtheorem{exmp}{Example}[section]

\begin{document}
\label{title}
\title{RIS-assisted Physical Layer Security}

\author{
	\IEEEauthorblockN{Xudong Li\IEEEauthorrefmark{1}, Matthias Frey\IEEEauthorrefmark{2}, Ehsan Tohidi\IEEEauthorrefmark{1}\IEEEauthorrefmark{3}, Igor Bjelakovi\'c\IEEEauthorrefmark{1}\IEEEauthorrefmark{3}, S\l{}awomir Sta\'nczak\IEEEauthorrefmark{1}\IEEEauthorrefmark{3}
	}\\
	\IEEEauthorblockA{\IEEEauthorrefmark{1}Technische Universit\"{a}t Berlin,} 
    \IEEEauthorblockA{\IEEEauthorrefmark{2}The University of Melbourne,}
    \IEEEauthorblockA{\IEEEauthorrefmark{3}%
    Fraunhofer Heinrich Hertz Institute}
}

\maketitle

\begin{abstract}
We propose a reconfigurable intelligent surface (RIS)-assisted wiretap channel, where the RIS is strategically deployed to provide a spatial separation to the transmitter, and orthogonal combiners are employed at the legitimate receiver to extract the data streams from the direct and RIS-assisted links. Then we derive the achievable secrecy rate under semantic security for the RIS-assisted channel and design an algorithm for the secrecy rate optimization problem. The simulation results show the effects of total transmit power, the location and number of eavesdroppers on the security performance.

\end{abstract}

\section{Introduction}
    The rapid advancement of wireless communication technologies has driven the exploration of new methods to improve coverage and optimize signal propagation. One such innovation is the deployment of reconfigurable intelligent surfaces (RISs) \cite{pan2021reconfigurable}, which are engineered surfaces capable of dynamically controlling electromagnetic waves. By altering how these waves are reflected, refracted, or absorbed, a RIS can extend coverage in challenging environments and shape the communication channels to improve signal quality and reliability \cite{huang2019reconfigurable}. 

    Recent studies have demonstrated that the RIS has the potential to significantly enhance the physical layer security (PLS) of wireless communications by directing signals toward legitimate users and away from potential eavesdroppers \cite{feng2020large}. Indeed, the secrecy rate can be greatly increased with the assistance of RISs \cite{chen2019intelligent}. However, the choice of security measure in PLS has a significant impact on practical relevance of the results \cite{bloch2011physical}.

    In the context of PLS, the wiretap channel model \cite{wyner1975} provided the initial framework for securing communications against eavesdroppers by minimizing information leakage. Weak secrecy was introduced to measure the leakage from an information theoretic perspective. Later, a more stringent measure called strong secrecy was proposed \cite{maurer1994strong}. The drawback of these security measures is that they do not quantify how much inference about the transmitted message the eavesdropper can make based on its received signal. The security measure that eliminates this drawback is semantic security, which has been introduced in \cite{bellare2012semantic} and has a clear operational interpretation\cite{frey2017mac}. 

    In this paper, we consider a wiretap channel with one RIS and several eavesdroppers. The legitimate receiver is able to distinguish signals from the direct link and the RIS-assisted link with the help of spatial separation and orthogonal combiners (SSOC) \cite{1033686}. We explore the security performance in this regime using semantic security as the secrecy measure. 

    \subsection{Prior work}
    RISs were shown to be able to sense the wireless environment and accordingly adjust its reflection coefficients dynamically to improve wireless communication \cite{pan2021reconfigurable}. Following these advancements, the use of RISs was explored to create a programmable wireless environment for PLS \cite{chen2019intelligent}, particularly in scenarios where traditional methods struggled to maintain robust security in complex environments \cite{shen2019secrecy,rafieifar2021secrecy}. In \cite{cui2019secure, zhou2021secure}, different algorithms for secrecy rate maximization by jointly optimizing total power constraint and RIS configuration were developed and a closed-form solution with fixed RIS configuration was given in \cite{5485016}. However, the performance gain of integrating RIS into PLS was not fully exploited. The achievable secrecy rate used in those algorithms was derived under strong the security criterion. Moreover, the algorithms were to optimize the total transmission power because the signals from the direct link and the RIS-assisted link were simply summed at the legitimate receiver.

    \subsection{Main Contribution and Outline}
    The main contribution of this paper can be summarized as follows:
    \begin{itemize}
        \item We introduce a vector-valued model for the RIS-assisted wiretap channel with SSOC. 
        \item We derive an achievable secrecy rate in this regime under the semantic security criterion, which is more stringent and operationally meaningful than strong secrecy used in prior works on RIS-assisted wiretap channels.
        \item We design a power allocation algorithm for the direct and RIS-assisted links that optimizes the secrecy rate and we evaluate the security enhancement.
    \end{itemize}
    In Section \ref{Sec 2. semantc}, we give the basic definitions. In Section \ref{Sec. 3}, we introduce the system model that is considered in this paper. Section \ref{Sec. Main results} contains the theorem of achievable secrecy rate and the proofs. We propose an optimization algorithm in Section \ref{Sec. Evaluation} and present simulation results in Section \ref{Sec. Simulation}. Section \ref{Sec. Conclusion} presents the paper's conclusions.
    
\section{Semantic Security} \label{Sec 2. semantc}
In this section, we provide the definition and interpretation of semantic security. The concept originates from the field of cryptography \cite{GOLDWASSER1984270}, but has also been extended to and used in the realm of physical layer security \cite{bellare2012semantic}.

We recall basic notions of wiretap channels and wiretap codes, which will be used in the definition of semantic security.
\begin{mydef}\label{def. Channel}
A wiretap channel
$T_{w} =((\mathcal{X},\mathcal{G}),(\mathcal{Y},\mathcal{F}),(\mathcal{Z},\\\mathcal{H}), K_{AB}, K_{{AE}})$ is defined by 
\begin{itemize}
    \item the measurable spaces $(\mathcal{X}, \mathcal{G})$ consisting  of input alphabet $\mathcal{X}$ with a $\sigma$-algebra $\mathcal{G}$ for the transmitter (Alice), $(\mathcal{Y}, \mathcal{F})$ representing the output alphabet with a $\sigma$-algebra $\mathcal{F}$ for the legitimate receiver (Bob), and $(\mathcal{Z},\mathcal{H})$ being the output alphabet and corresponding $\sigma$-algebra for the eavesdropper (Eve),
    \item the channel $K_{AB}$ from Alice to Bob defined by a  Markov kernel on $\mathcal{F}\times \mathcal{X}$ (i.e. a mapping $K_{AB}: \mathcal{F} \times \mathcal{X} \rightarrow [0,1]$ such that a) for each $x\in \mathcal{X}$ the map $\mathcal{F}\ni F \mapsto K_{AB}(F, x)$ is a probability measure and b) for each $F\in \mathcal{F}$ the map $\mathcal{X}\ni x \mapsto K_{AB}(F,x)$ is $\mathcal{G}$-measurable),
   \item the channel $K_{AE}$ from Alice to Eve given as a Markov kernel on $\mathcal{H} \times \mathcal{X}$ (i.e. a mapping $K_{AE}: \mathcal{H} \times \mathcal{X} \rightarrow [0,1]$ with properties that  are analogous to those of $K_{AB}$ from the previous item).
\end{itemize}
\end{mydef}

\begin{mydef}\label{def. Wiretap Code}
    An $(n,\mathcal{M})$-wiretap code $(\mathcal{C}_n, D_n)$ for the channel $T_{w}$ on the alphabet $\mathcal{X}$ consists of
    \begin{itemize}
        \item a randomized encoder $\mathcal{C}_n: \mathcal{M} \times \mathcal{W} \to \mathcal{X}^{n}$, where $\mathcal{M}$ is the set of all messages $\{1, \dots, L\}$ and $W$ is a random variable with values in $\mathcal{W}=\{1, \dots, L_1\}$ distributed according to a distribution $P_W$ on $\mathcal{W}$.
        In order to encode each message $m \in \mathcal{M}$, we perform a random experiment $W$ with outcomes $w \in\mathcal{W}$ to generate a sequence $x^n(m, w) \coloneq \mathcal{C}_n(m,w) \in \mathcal{X}^{ n}$.
        \item a (deterministic) decoder $D_n: \mathcal{Y}^{n} \to \mathcal{M}$.
    \end{itemize}
\end{mydef}

The information rate $R$ and the randomness rate $R_1$ are described by
\begin{align} 
    &R = \frac{1}{n}\log{|\mathcal{M}|} = \frac{1}{n}\log{L}\\ 
    &R_1 = \frac{1}{n}\log{|\mathcal{W}|} = \frac{1}{n}\log{L_1}.
\end{align}

By a wiretap codebook of block length $n \geq 1$, information rate $R \geq 0$ and randomness rate $R_1 \geq 0$, we mean the finite sequence $\{\mathcal{C}_n(m,w): m \in \mathcal{M}, w \in \mathcal{W}\}$ and $\mathcal{C}_n(m,w) \in \mathcal{X}^n$ is the transmitted codeword on the input alphabet. We denote the codebook by the same symbol $\mathcal{C}_n$.

For the definition of semantic security we assume that the goal of Eve is represented by a partition $\Pi$ of the message set $\mathcal{M}$ and her task is to determine the partition element $\pi \in \Pi$ that the transmitted message $m\in \mathcal{M}$ belongs to. Two typical goals are given in the next two examples.

\begin{exmp}
    If Eve aims to reconstruct the first bit of a binary message, the message set $\mathcal{M}$ will be partitioned into two subsets, one for messages starting with 0 and the other for those starting with 1.  
\end{exmp}

\begin{exmp}
    If Eve aims to reconstruct the entire message, the message set will be partitioned into singleton sets, i.e., $\Pi = \{\{m\}: m \in \mathcal{M}\}$.
\end{exmp}
The idea behind the definition of semantic security is to compare Eve's success probability in determining the correct partition element with pure guessing. A semantically secure communication scheme should have a small gap between these success probabilities for every possible partition of the message space.

\begin{mydef}\label{def. Semantic}
    The eavesdropper's maximum advantage in reconstructing the transmitted message under semantic security \cite{bellare2012semantic} criterion is given by
    \begin{align}\label{def:semantic-advantage-partition}
            Adv_{ss}(\mathcal{C}_n, Z^n) \coloneq \max_{P_M, \Pi} &\left(\sup_{g :\mathcal{Z}^n \to \Pi} P_{M,Z^n}(M\in g(Z^n))\right. \nonumber\\ 
            &-\left.\max_{\pi \in \Pi} P_{M}(M\in \pi) \right),
    \end{align}
where $P_M$ and $\Pi$ denote the distribution of the message $M$ and a partition of the message set $\mathcal{M}$ respectively, the function $g :\mathcal{Z}^n \to \Pi $ models the selection of a partition element of $\Pi$ based on the observation of Eve's channel output $Z^n$, and $P_{M,Z^{n}}$ denotes the joint distribution of the message and Eve's channel output.

    We say that the codebook $\mathcal{C}_n$ achieves $\delta-$semantic security, for $\delta > 0$, if 
    \[
    Adv_{ss}(\mathcal{C}_n, Z^n) \leq \delta.
    \]
\end{mydef}
We interpret the semantic security advantage defined in \eqref{def:semantic-advantage-partition} as follows: The term $\max\limits_{g :\mathcal{Z}^n \to \Pi}\! P_{M,Z^n}(M\!\in\! g(Z^n))$ denotes the maximum success probability for determining the partition element of $\Pi$ upon observation of Eve's channel output. This implies that Eve's objective is defined by a partition $\Pi$ of $\mathcal{M}$, and Eve's goal is to guess to which partition element the transmitted message $M$ belongs. The term $\max\limits_{\pi \in \Pi} P_{M}(M\in \pi)$ represents the most likely partition element under the distribution $P_M$ without knowledge of Eve's channel output, i.e., the optimal guess. Therefore, the right-hand side of \eqref{def:semantic-advantage-partition} gives the largest advantage over pure guessing Eve can achieve based on the knowledge of $P_M$ and the choice of partition $\Pi$. If the codebook achieves $\delta-$semantic security, then Eve's maximum success probability  is $\delta-$close to pure guessing. 

The results in \cite[Proposition 1]{ling2014semantically} show that for the case of wireless channels, semantic security implies both the weak secrecy introduced in\cite{wyner1975} and strong secrecy from \cite{maurer1994strong} which have been used traditionally in the realm of physical layer security. Apart from this, semantic security has the advantage, in contrast to weak and strong secrecy, that it is operationally defined. This means that it delivers a quantitative statement about the maximum probability of success compared to pure guessing for a clearly defined class of eavesdropping attacks. 

In the context of secure communication over noisy channels, achieving semantically secure communication involves not only guaranteeing confidentiality but also reliability. The reliability of communication is evaluated by the average decoding error probability.  
\begin{mydef}\label{def. error}
    The average decoding error probability $P_{\mathcal{E}}(\mathcal{C}_n,D_n)$ for an  $(n,\mathcal{M})$-wiretap code is defined as
    \begin{align}        
    P_{\mathcal{E}}(\mathcal{C}_n,D_n) = &\frac{1}{LL_1} \sum_{w=1}^{L_1} \sum_{m=1}^L  \nonumber\\&P_{{Y}^n|{X}^n}(D({Y}^n) \neq m|{X}^n={x}^n(m,w))\label{eq. average decoding error},
    \end{align}
    where $P_{{Y}^n|{X}^n}$ denotes the conditional probability of legitimate receiver's channel output given the transmitted codeword.
\end{mydef}

We can now define the achievable secrecy rate that quantifies the rate at which reliable and secure communication is achieved.
\begin{mydef}\label{def. Achievable}
    A non-negative number $R_s$ is called an achievable secrecy rate under semantic security for a wiretap channel if there is a sequence $(\mathcal{C}_n,D_n)_{n \in \mathbb{N}}$ of $(n,\mathcal{M})$-codes such that for every $\epsilon_1, \epsilon_2,\epsilon_3 > 0$ there is $n_0 \in \mathbb{N}$ such that for all $n \geq n_0$,
\begin{equation*}
 P_{\mathcal{E}}(\mathcal{C}_n,D_n) \!<\! \epsilon_1, \
    Adv_{ss}(\mathcal{C}_n, {Z}^n) < \epsilon_2, \ \text{and}  \
    \frac{1}{n} \log L \geq R_s-\epsilon_3.
\end{equation*}
\end{mydef}

\section{System Model} \label{Sec. 3}
\subsection{Wireless Communication Model}\label{Subsec. Wireless}
We consider a RIS-assisted communication system as shown in Fig. \ref{fig. wireless}, where Alice equipped with $N_t$ antennas intends to send confidential information to Bob in the presence of a finite number of Eves, and Bob as well as the Eves are equipped with $N_r$ antennas. To enhance the security of the transmission, a RIS consisting of $N$ reflecting elements is deployed. Notations $\bH_{AB} \in \mathbb{C}^{N_r\times N_t}$, $\bH_{AR} \in \mathbb{C}^{N\times N_t}$ and $\bH_{RB} \in \mathbb{C}^{N_r\times N}$ are used for the direct channel between Alice and Bob, the channel from Alice to the RIS, and the channel from the RIS to Bob, respectively, whereas  $\bH_{AE_j} \in \mathbb{C}^{N_r\times N_t}, \bH_{RE_j}\in \mathbb{C}^{N_r\times N}$ denote the direct channel between Alice and Eve$_j$, and between the RIS and Eve$_j$, $j\in\{1,\dots,d\}$, respectively.

\begin{figure}[tb]
    \vspace{-0.35cm}
    \hspace{-1.5cm}
    \centering
    \includegraphics[width=0.4\textwidth]{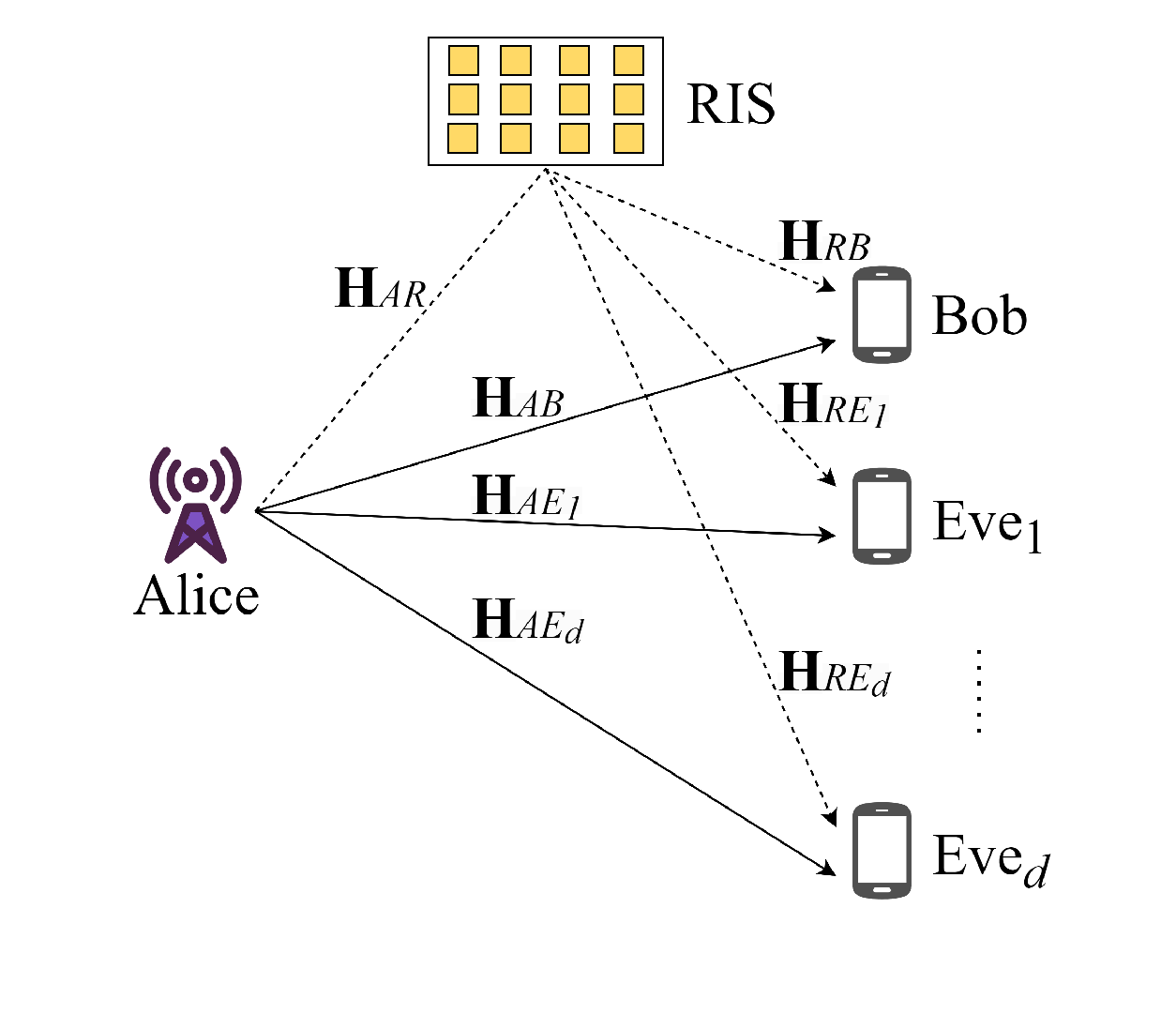}
    \vspace{-0.7cm}
    \caption{RIS-assisted secure communication system.}
    \vspace{-0.3cm}
    \label{fig. wireless}
\end{figure}

Utilizing both direct and RIS-assisted links, we aim to transmit two data streams, denoted by $X_1$ and $X_2$, each with a corresponding precoder $\bp_k \!\in\! \mathbb{C}^{N_t \times 1}, k\in\{1,2\}$. The random variables $X_1, X_2 \in \mathbb{C}$ over the set $\mathcal{X}$ denote the transmitted data for a single channel use. The RIS performs a linear mapping from the incident signal vector to a reflected signal vector based on an equivalent $N \times N$ diagonal phase-shift matrix denoted by $\boldsymbol{\theta} \in \mathbb{C}^{N \times N}$. Consequently, the received signals at Bob and Eve$_j$ are given by
\begin{align} 
    Y &= (\bH_{AB} + \bH_{RB}\btheta\bH_{AR}) (\bp_1 X_1 + \bp_2 X_2) + E_{AB},\label{eq.origRec}\\
    Z_j &= (\bH_{AE_j} + \bH_{RE_j}\btheta\bH_{AR}) (\bp_1 X_1 + \bp_2 X_2) + E_{AE_j},\label{eq.origRecEve}    
\end{align}
where $E_{AB},E_{AE_j} \in \mathbb{C}^{N_r \times 1}$ represent complex additive white Gaussian noises (AWGNs) at Bob and Eve$_j$, respectively.

In this paper, we maintain the separation of the two transmit data streams via spatial isolation using a beam space representation. To achieve this, we assume that the RIS is strategically deployed in a location that enables such spatial separation. In other words, to provide distinct paths, the difference in the path directions should exceed the spatial resolution \cite{1033686}. Having two distinct paths allows us to design the precoder $\bp_1$ to maximize the gain for the direct link while ensuring it lies in the null space of the RIS-assisted link, and vice versa for precoder $\bp_2$ \cite{1033686}. Consequently, we have $\bH_{RB}\btheta\bH_{AR}\bp_{1} = \bH_{AB}\bp_{2} = 0$ and hence the received signal at Bob in \eqref{eq.origRec} simplifies to

\begin{equation}
\begin{aligned}  
    Y &= \bH_{AB} \bp_1 X_1 + \bH_{RB}\btheta\bH_{AR} \bp_2 X_2 + E_{AB}.
    \label{eq.simpRec}
\end{aligned} 
\end{equation}

Similar to the transmit side, at receiver Bob, we employ a set of orthogonal combiners $\mathbf{v}_{B,k}$ to extract $X_k$ for $k \in \{1,2\}$ \cite{1033686}. Given the spatial separation and an appropriate design of the combiners, we have $\bv_{B,1}^H\bH_{RB}\btheta\bH_{AR} \bp_2 = \bv_{B,2}^H\bH_{AB} \bp_1 = 0$. Hence we can express the extraction $Y_1 \coloneq \bv_{B,1}^H Y$ and $Y_2 \coloneq \bv_{B,2}^H Y$ as
    \begin{equation}
    \begin{aligned}
        Y_{1} &= \bv_{B,1}^H\bH_{AB} \bp_1 X_1 + \bv_{B,1}^H E_{AB},\\
        Y_{2} &= \bv_{B,2}^H\bH_{RB}\btheta\bH_{AR} \bp_2 X_2 + \bv_{B,2}^H E_{AB}.\\
    \end{aligned}        
    \label{eq.combineRec1}
    \end{equation}

We define $\alpha_{AB,1} \coloneq \bv_{B,1}^H\bH_{AB} \bp_1$ to be the gain of the direct channel between Alice and Bob, and $\alpha_{AB,2}^{(\boldsymbol{\theta})} \coloneq \bv_{B,2}^H\bH_{RB}\btheta\bH_{AR} \bp_2$ to be the gain of RIS-assisted channel. The superscript $\boldsymbol{\theta}$ emphasizes that the configuration of the RIS has an impact on the gain. We denote $E_{AB,1}\coloneq \bv_{B,1}^H E_{AB}$ and $E_{AB,2} \coloneq \bv_{B,2}^H E_{AB}$ as complex AWGNs for the corresponding channels. Consequently, \eqref{eq.combineRec1} leads to two independent components:
\begin{equation} 
    \begin{aligned}
        Y_{1} &= \alpha_{AB,1} X_1 + E_{AB,1},\\
        Y_{2} &= \alpha_{AB,2}^{(\boldsymbol{\theta})} X_2 + E_{AB,2}.\\
    \end{aligned}        
    \label{eq.combineRec2}
    \end{equation}
    
This independence arises from the orthogonality of $\bv_{B,1}$ and $\bv_{B,2}$, i.e., $\bv_{B,1}^H\bv_{B,2}=0$.

For the Eves, we do not assume the use of orthogonal combiners which means that our security guarantees hold regardless of how the Eves choose to process their received signals. Hence, the received signal at Eve$_j$ remains as in \eqref{eq.origRecEve}.

\begin{Remark}
In this paper, we consider Bob equipped with multiple antennas. By leveraging the spatial separation inherent in the direct and RIS-assisted links and utilizing suitable combiners, we can effectively recover the transmitted data at Bob. This setting represents one way of establishing two independent communication paths between Alice and Bob. A similar method applies to other technologies that provide spatial separation such as distributed multiple-input multiple-output (MIMO) technique.
\end{Remark}

\subsection{RIS-assisted Wiretap Channel}
In the following part of the paper, we abstract the earlier-mentioned details concerning parameter settings at Alice, the RIS, and Bob, presuming they are appropriately configured to ensure the SSOC. Therefore, we focus on \emph{determining achievable communication rates} under the constraint of semantic security, while abstracting the technical assumptions from the preceding section.
We also extend the wireless channel model above from single channel use to multiple channel uses in order to investigate the relationship between the length of the code words, the quality of the security level, and the reliability level \cite{bloch2011physical}.

We propose a RIS-assisted wiretap system model with SSOC as shown in Fig. \ref{fig.1}, which consists of one Alice, one Bob,  multiple Eves (Eve$_1$, $\dots$, Eve$_d$), and one RIS. Our objective is to send messages reliably between the legitimate parties while keeping all Eves ignorant of messages in the sense of achievable $\delta-$semantic security (cf. Definition \ref{def. Semantic}), suitably small $\delta>0$, with the assistance of a RIS. 

\begin{figure}[tb]
    \centering
    \includegraphics[width=0.5\textwidth]{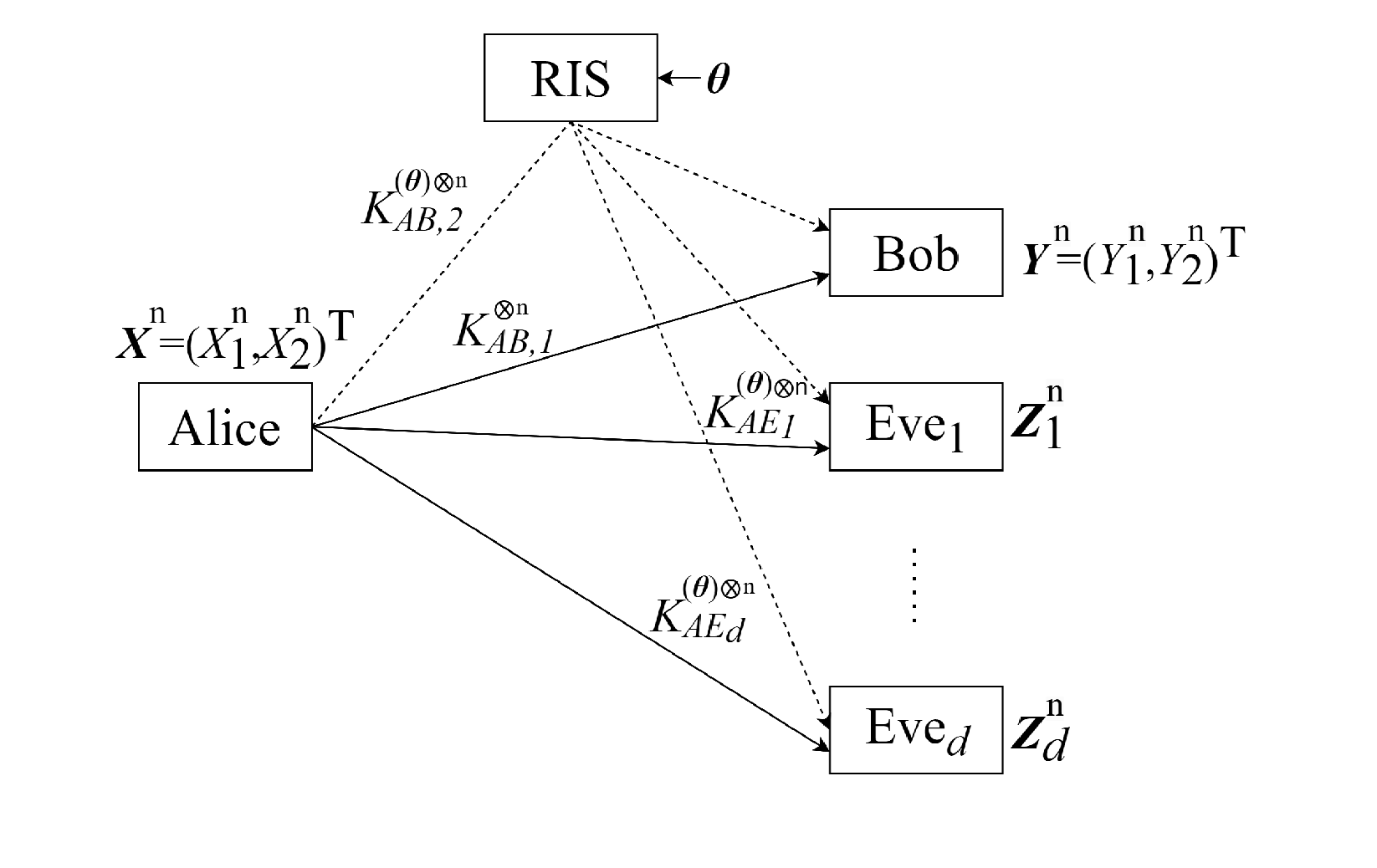}
    \vspace{-0.8cm}
    \caption{RIS-assisted wiretap channel model with SSOC.}
    \label{fig.1}
\end{figure}

We can extend the classic wiretap channel from Definition \ref{def. Channel} to a RIS-assisted wiretap channel as $T_{\text{RIS}} =((\mathcal{X},\mathcal{G}),(\mathcal{Y},\mathcal{F}),(\mathcal{Z}_j,\mathcal{H}_j), K^{(\boldsymbol{\theta})}_{AB}, K^{(\boldsymbol{\theta})}_{{AE}_j},\boldsymbol{\theta})$ for $j=1,\dots, d$. $\mathcal{Z}_j$ is the output alphabet with a $\sigma$-algebra $\mathcal{H}_j$ for Eve$_j$. 
The Markov kernel $K^{(\boldsymbol{\theta})}_{AB}$ denotes the product of $K_{AB,1}$ and $K^{(\boldsymbol{\theta})}_{AB,2}$, i.e. $K^{(\boldsymbol{\theta})}_{AB}(A_1 \times A_2, \boldsymbol{x})=K_{AB,1}(A_1,x_1)\times K^{(\boldsymbol{\theta})}_{AB,2}(A_2, x_2)$ for $x_1,x_2 \in \mathcal{X}$, $A_1,A_2 \in \mathcal{F}$, where $K_{AB,1},K^{(\boldsymbol{\theta})}_{AB,2}: \mathcal{F} \times \mathcal{X} \rightarrow [0,1]$. This is a generalization of the model \eqref{eq.combineRec2}. Based on the absence of assumptions about Eves, we generalize the Markov kernel between Alice and RIS for Eve$_j$ as $K^{(\boldsymbol{\theta})}_{{AE}_j}: \mathcal{H}_j \times \mathcal{X}^{2} \rightarrow [0,1]$. The $n$-fold product of $K(A,x)$ is defined as $K^{\otimes n}(\times^{n}_{i=1} A_i, x^n) = \prod^n_{i=1} K(A_i,x_i)$. The RIS is configured by the vector of parameters $\boldsymbol{\theta}\in \Theta$, which has an impact on the main channel and eavesdropper channels. 

Let $\boldsymbol{X}^n = [X^n_1, X^n_2]^T$ be a random vector, where $X_1, X_2$ are the random variables over the set $\mathcal{X}$ for the direct channel input and the RIS-assisted channel input, respectively. The superscript $n$ denotes the number of channel uses. 
With proper spatial separation and appropriate orthogonal combiners, the output of channel $K^{(\boldsymbol{\theta}){\otimes n}}_{AB}$ at Bob can be split into two vector components: $\boldsymbol{Y}^n = [Y^n_1, Y^n_2]^T$, where $Y_1, Y_2$ are the random variables over the set $\mathcal{Y}$ for the corresponding channel outputs. An input distribution $Q_{\boldsymbol{X}}$ induces an output distribution $Q_{\boldsymbol{Y}}$ via $Q_{\boldsymbol{Y}}(A) = \int K^{(\boldsymbol{\theta})}_{AB}(A,\boldsymbol{x})Q_{\boldsymbol{X}}(d\boldsymbol{x})$ for $A \in \mathcal{F}$. The output of channel $K^{(\boldsymbol{\theta}){\otimes n}}_{AE_j}$ at Eve$_j$ is denoted by $\boldsymbol{Z}_j^n$ over the set $\mathcal{Z}_j^{n}$. The induced output distribution is $Q_{\boldsymbol{Z}_j}(A) = \int K^{(\boldsymbol{\theta})}_{AE_j}(A,\boldsymbol{x})Q_{\boldsymbol{X}}(d\boldsymbol{x})$ for $A \in \mathcal{H}_j$. The n-products of corresponding distributions are denoted by $Q_{\boldsymbol{X}}^n$, $Q_{\boldsymbol{Y}}^n$ and $Q_{\boldsymbol{Z}_j}^n$.

The wiretap code for the channel $T_{\text{RIS}}$ is defined analogous to Definition \ref{def. Wiretap Code}. The difference is that for every $\boldsymbol{\theta}\in \Theta$, we have a randomized encoder $\mathcal{C}_n^{(\boldsymbol{\theta})}$ with vector output on the alphabet $\mathcal{X}^{2\times n}$ and a (deterministic) decoder $D_n$ with vector input on the alphabet $\mathcal{Y}^{2\times n}$. The decoder $D_n$ is not directly related to $\boldsymbol{\theta}$, but is related to the channel information, which is affected by $\boldsymbol{\theta}$. We will drop the superscript $(\boldsymbol{\theta})$ in $\mathcal{C}_n^{(\boldsymbol{\theta})}$ to simplify the notation.

In order to encode $m \in \mathcal{M}$, we perform a random experiment with $W$ according to $P_W$ to generate a sequence $\boldsymbol{x}^n(m,w) \in \mathcal{X}^{2\times n}$. The sequence $\boldsymbol{x}^n (m, W) \in \mathcal{C}_n$ (corresponding to random selection of $w \in \mathcal{W}$) is chosen as the input for the channel. Then the induced output distribution of $\boldsymbol{Z}_j^n$ for each message $m \in \mathcal{M}$ by a given codebook $\mathcal{C}_n$ is

\begin{equation}
    P_{\boldsymbol{Z}_j^n|\boldsymbol{X}^n(m,\cdot)} \coloneq
      \sum_{w=1}^{L_1} P_W(w) P_{\boldsymbol{Z}_j^n|\boldsymbol{X}^n(m,w)}.
\end{equation}

\section{Main Results} \label{Sec. Main results}
In this section, we first introduce a measure of secrecy which is equivalent to semantic security but which is easier to work with analytically. Then we explicitly derive an upper bound on achievable secrecy rates for the RIS-assisted wiretap channel by following the line of arguments about the channel resolvability from \cite{frey2018resolvability}.

\subsection{Distinguishing Security}
 Since it is not easy to prove semantic security for a particular wiretap channel directly, we introduce an equivalent but analytically more accessible notion of secrecy. Before this, we define the underlying distance measure.

\begin{mydef}[Total Variation Distance \protect{\cite[5.24]{ccinlar2011probability}}]
 \label{def. TV distance}
    The total variation distance between two probability measures $P$ and $Q$ defined on the measurable space $(\Omega,\mathcal{F})$ is given by
    \begin{equation}
        ||P- Q||_{\rm{TV}} = \sup_{A \in \mathcal{F}}|P(A)-Q(A)|.
    \end{equation}
\end{mydef}

\begin{mydef}\label{def. Distinguishing Security}
     The eavesdropper's maximum advantage in reconstructing the transmitted message under distinguishing security \cite{bellare2012semantic} criterion is given by 
     \begin{equation}
            Adv_{ds}(\mathcal{C}_n, Z^n) 
            \coloneq \max_{m_1,m_2 \in \mathcal{M}} ||P_{Z^n|m_1} - P_{Z^n|m_2}||_{\rm{TV}}\label{distinguishing TV},
     \end{equation}
     where $\mathcal{M}$ is the message set and $P_{Z^n|m_1},P_{Z^n|m_2}$ are the output distribution given messages $m_1,m_2$, respectively.
    
     We say that the codebook $\mathcal{C}_n$ achieves $\delta-$distinguishing security, for $\delta > 0$, if 
    \[
     Adv_{ds}(\mathcal{C}_n, Z^n) \leq \delta.
    \]
\end{mydef}

\begin{theorem}\label{theo. Equivalence}
\begin{itemize}
    \item[] 
    \item Semantic security and distinguishing security are equivalent, \!i.e.,\! the following inequalities are valid:
    \begin{equation*}
        Adv_{ss} (\mathcal{C}_n, Z^n) \leq Adv_{ds} (\mathcal{C}_n, Z^n) \leq 2 \cdot Adv_{ss}(\mathcal{C}_n, Z^n). 
    \end{equation*}
    \item For any probability measure $P'$ on $Z^n$, we have
    \begin{equation*}
         Adv_{ds}(\mathcal{C}_n, Z^n) \leq 2 \cdot \max_{m\in \mathcal{M}} ||P_{Z^n|m} - P'||_{\rm{TV}}.
    \end{equation*} 
\end{itemize}
\end{theorem}
The first item follows from \cite[Theorem 1]{bellare2012semantic} and the second item follows from the triangle inequality.

Theorem \ref{theo. Equivalence} helps us to obtain the semantic secrecy results by proving Theorem \ref{Distinguishing achievablity}. 

\subsection{Distinguishing and Semantic Security Achievability}
Before stating the theorem on achievable secrecy rate, we provide a definition of the information density \cite[Eq. (13)]{goldfeld2016semantic} which is adapted to our case. For all $(\boldsymbol{x}^{n},\boldsymbol{z}_j^{n}) \in \mathcal{X}^{2\times n} \times \mathcal{Z}_j^{2\times n}$ and the Markov kernel $K^{(\boldsymbol{\theta})}_{AE_j}$, $j= 1,\dots,d.$, the information density is given by
\begin{equation}
    i[\boldsymbol{x}^{n};\boldsymbol{z}_j^{n}] = \log\left(\frac{dK^{(\boldsymbol{\theta}){\otimes n}}_{AE_j}(\ \cdot \ ,\boldsymbol{x}^n)}{dQ_{\boldsymbol{Z}_j}^n(\cdot)}(\boldsymbol{z}_j^n)\right), \label{eq. def. information density}
\end{equation} 
where the derivative on the right-hand side is the Radon-Nikodym derivative \cite[5.18]{ccinlar2011probability} of the measure $K^{(\boldsymbol{\theta}){\otimes n}}_{AE_j}(\ \cdot \ ,\boldsymbol{x}^n)$ with respect to measure $Q_{\boldsymbol{Z}_j}^n(\cdot)$. The Radon-Nikodym derivative exists for $Q_{\boldsymbol{X}}^n$-a.e. $\boldsymbol{x}$ due to the following argument:\\
For $A \in \mathcal{H}_j$ we have $Q_{\boldsymbol{Z}_j}^n(A) = \int K^{(\boldsymbol{\theta})\otimes n}_{AE_j}(A,\boldsymbol{x})Q^n_{\boldsymbol{X}}(d\boldsymbol{x})$. If $Q_{\boldsymbol{Z}_j}^n(A) = 0$, this implies  $K^{(\boldsymbol{\theta})\otimes n}_{AE_j}(A,\boldsymbol{x}) = 0$ for $Q_{\boldsymbol{X}}^n$-a.e. $\boldsymbol{x}$. Consequently $K^{(\boldsymbol{\theta})\otimes n}_{AE_j}(\cdot,\boldsymbol{x})$ is absolutely continuous with respect to $Q_{\boldsymbol{Z}_j}^n$, in symbols $K^{(\boldsymbol{\theta})\otimes n}_{AE_j}(\cdot,\boldsymbol{x}) \ll Q_{\boldsymbol{Z}_j}^n$ for $Q_{\boldsymbol{X}}^n$-a.e. $\boldsymbol{x}$. Thus the Radon-Nikodym theorem \cite[Theorem 5.11]{ccinlar2011probability} states the existence of the Radon-Nikodym derivative with the condition above. This is also true for the existence of the Radon-Nikodym derivatives in $i[\boldsymbol{x}^{n};\boldsymbol{y}^{n}]$ for almost all $(\boldsymbol{x}^{n},\boldsymbol{y}^{n}) \in \mathcal{X}^{2\times n} \times \mathcal{Y}^{2\times n}$ and the Markov kernel $K^{(\boldsymbol{\theta})\otimes n}_{AB}$.

\begin{theorem}[Achievable Secrecy Rate under Distinguishing and Semantic Security]\label{Distinguishing achievablity}
    Suppose $((\mathcal{X},\mathcal{G}),(\mathcal{Y},\mathcal{F}),(\mathcal{Z}_j,\mathcal{H}_j), K^{(\boldsymbol{\theta})}_{AB},K^{(\boldsymbol{\theta})}_{{AE}_j},\boldsymbol{\theta})$ is a RIS-assisted wiretap channel ($j=1,\dots, d$), $Q_{\boldsymbol{X}}$ is the input distribution, $Q_{\boldsymbol{X},\boldsymbol{Y},\boldsymbol{Z}}$ is the joint distribution induced by $Q_{\boldsymbol{X}},K^{(\boldsymbol{\theta})}_{AB}$ and $K^{(\boldsymbol{\theta})}_{{AE}_j}$. Assume that:
    \begin{itemize}
        \item the moment generation functions of $i[\boldsymbol{X};\boldsymbol{Y}]$ and $i[\boldsymbol{X};\boldsymbol{Z}_j]$, $j= 1,\dots,d$, exist and are finite for some $t > 0$, i.e.
        $\mathbb{E}[\exp{(ti[\boldsymbol{X};\boldsymbol{Y}])}] < +\infty$ and $\mathbb{E}[\exp{(ti[\boldsymbol{X};\boldsymbol{Z}_j])}] < +\infty$ for $j= 1,\dots,d$.
        \item $R_s<I[\boldsymbol{X};\boldsymbol{Y}]-\max\limits_{j=1,\dots,d} I[\boldsymbol{X};\boldsymbol{Z}_j]$.
    \end{itemize}
    
    Then there exist $\xi_j>0$, $j= 1,\dots,d$, $\lambda_1,\lambda_2>0$ such that for  sufficiently large $n\in\mathbb{N}$ we have\\
    $||P_{\boldsymbol{Z}_j^n|\boldsymbol{X}^n(m,\cdot)} - Q_{\boldsymbol{Z}_j}^n||_{\rm{TV}} \leq \exp\{-n\xi_j\}$ (Distinguishing Security) and 
    $P_{\mathcal{E}}(\mathcal{C}_n,D_n) \leq \exp\{-n\lambda_2\}$ (Reliability)
    with probability at least $1-\exp\{-n\lambda_1\}$.
    
    Consequently, by Theorem \ref{theo. Equivalence}, we obtain for the semantic security advantage:
    \begin{equation*}\label{eq. semantic-theorem}
        Adv_{ss}(\mathcal{C}_n, \boldsymbol{Z}_j^{n}) \leq 2\cdot \exp\{-n\xi_j\},  \quad j=1, \ldots ,d.      \end{equation*}  
\end{theorem}

To prove Theorem \ref{Distinguishing achievablity}, we derive the upper bounds for the distinguishing security part and reliability part separately. In the end, both parts are combined into a single statement, concluding the proof. Some auxiliary results for the proof are provided in Appendix \ref{Sec. Auxiliary Results}.

\subsubsection{Distinguishing Security}
The proof of distinguishing security relies on an application of McDiarmid's inequality which is given as follows:
\begin{theorem}[McDiarmid Inequality \cite{mcdiarmid1989method}] \label{theo. McDiarmid Inequality} 
    Let $X_1,\dots,X_n$ be independent random variables taking values in (measurable) sets $\mathcal{V}_1, \ldots , \mathcal{V}_n$, and $\boldsymbol{x}_{p}^n=(x_1,\dots,x_p,\dots,x_n), \boldsymbol{x}_{p}'^n=(x_1,\dots,x'_p,\dots,x_n) \in \mathcal{V}_1 \times \cdots \times \mathcal{V}_n$ pairs of vectors which differ only in the $p$-th component, $p=1,\ldots , n$. Suppose that for a given (measurable) function $f:\mathcal{V}_1 \times \cdots \times \mathcal{V}_n \to \mathbb{R}$ there exists $c\in \mathbb{R}$ such that
    \begin{equation}\label{bounded-differences}
        |f(\boldsymbol{x}_{p}^n)-f(\boldsymbol{x}_{p}'^n)| \leq c,
    \end{equation}
    for all $p=1, \ldots ,n$ and all $\boldsymbol{x}_{p}^n, \boldsymbol{x}_{p}'^{n} \in \mathcal{V}_1 \times \cdots \times \mathcal{V}_n$. Then for the random vector $\boldsymbol{X}^n=(X_1, \ldots , X_n)$ and any $\epsilon>0$ we have
    \begin{equation}\label{McDiarmid}
        \mathbb{P}(f(\boldsymbol{X}^n)-\mathbb{E}[f(\boldsymbol{X}^n)] \geq \epsilon) \leq \exp{\left\{\frac{-2{\epsilon}^2}{c^2n}\right\}}.
    \end{equation}
\end{theorem}
The key point of the proof is to derive an upper bound on the expectation value appearing in \eqref{McDiarmid}.
The functions of interest are given by
    \begin{equation}\label{f(C)}
        f_j(\mathcal{C}_n) \coloneq ||P_{\boldsymbol{Z}_j^n|\boldsymbol{X}^n(m,\cdot)} - Q_{\boldsymbol{Z}_j}^n||_{\rm{TV}}, \quad j=1,\dots,d,
    \end{equation} where
    \begin{equation}\label{eq. observed-output-measure}
        P_{\boldsymbol{Z}_j^n|\boldsymbol{X}^n(m,\cdot)} = \frac{1}{L_1} \sum_{w=1}^{L_1} K^{(\boldsymbol{\theta}){\otimes n}}_{{AE}_j}(\ \cdot \ , \boldsymbol{x}^n(m,w)).
    \end{equation}
and we consider a uniformly distributed randomness $W$.

In the first step, we show that \eqref{bounded-differences} in Theorem \ref{theo. McDiarmid Inequality}
holds for $c= \frac{1}{L_1}$. Let $\mathcal{C}_n$ and $\mathcal{C_n'}$ denote codebooks which differ only in codeword corresponding to arbitrary but fixed tuple $(m',w')$.  By using the triangle inequality of total variation distance \cite[Proposition 4.2.5]{athreya2006measure}, we obtain
\begin{align}
    &|f_j(\mathcal{C_n})-f_j(\mathcal{C'_n})| \nonumber\\
    \leq &\lVert \frac{1}{L_1} \sum_{w=1}^{L_1} K^{(\boldsymbol{\theta}){\otimes n}}_{{AE}_j}(\cdot, \boldsymbol{x}^n(m,w))\nonumber\\ 
    &- \frac{1}{L_1} \sum_{w=1}^{L_1} K^{(\boldsymbol{\theta}){\otimes n}}_{{AE}_j}(\cdot, \boldsymbol{x'}^n(m',w'))\rVert_{\rm TV}\nonumber\\
    = &\frac{1}{L_1}\lVert K^{(\boldsymbol{\theta}){\otimes n}}_{{AE}_j}(\cdot, \boldsymbol{x}^n(m,w))\nonumber\\
    &-K^{(\boldsymbol{\theta}){\otimes n}}_{{AE}_j}(\cdot, \boldsymbol{x'}^n(m',w'))\rVert_{\rm TV}\nonumber\\
    \leq &\frac{1}{L_1}, \label{eq. c bounded}
\end{align} 
where the last step follows by the fact that the total variation distance between any two probability measures is always less than or equal to 1.

In what follows we consider the random codebook $\mathcal{C}_n$,
where the random codewords $\boldsymbol{X}^n(m,w)$, $m\in\mathcal{M}, w\in \mathcal{W}$, are independent and distributed according to $Q_{\boldsymbol{X}}^n$.

As a consequence of \eqref{eq. c bounded} and Theorem \ref{theo. McDiarmid Inequality}, we obtain the following inequality for the random codebook and $j=1, \ldots ,d$:
\begin{equation}\label{eq. derived-McDiarmid}
\mathbb{P}_{\mathcal{C}_n}\left(f_j(\mathcal{C}_n)-\mathbb{E}_{\mathcal{C}_n}[f_j(\mathcal{C}_n)] \geq \epsilon_j\right) \leq \exp{\left\{-2{\epsilon_j}^2 L_1\right\}}.
\end{equation}

Our next goal is to derive an upper bound on the expectation $\mathbb{E}_{\mathcal{C}_n}(f_j(\mathcal{C}_n))$, $j=1,\ldots, d,$ appearing in \eqref{eq. derived-McDiarmid}.

\begin{lemma}[Bound for Expectation of TV Distance]\label{Bound for Expectation of TV Distance}
    Suppose that the moment generating functions of $i[\boldsymbol{X};\boldsymbol{Z}_j]$ exist and are finite for $j=1,\dots,d$ and some $t > 0$, and the randomness rate satisfies $R_1 > \max\limits_{j=1,\dots,d}I[\boldsymbol{X};\boldsymbol{Z}_j]$, then there exist  $\beta_{j} > 0$ such that
    \begin{equation}
       \mathbb{E}_{\mathcal{C}_n}(f_j(\mathcal{C}_n)) \!=\! \mathbb{E}_{\mathcal{C}_n}[||P_{\boldsymbol{Z}_j^n|\boldsymbol{X}^n\!(m,\cdot)} \!-\!Q_{\boldsymbol{Z}_j}^n||_{\rm{TV}}] \!\leq\! \exp\!{\{\!-n\beta_{j}\!\}},
    \end{equation} 
    for $j=1,\ldots ,d$ and all sufficiently large $n\in\mathbb{N}$.
\end{lemma}
Before proving Lemma \ref{Bound for Expectation of TV Distance} we state two auxiliary results that are used in the proof.

We first define the typical set $\mathcal{T}_{\epsilon_j}$ and the atypical set $\mathcal{T}_{\epsilon_j}^c$ \cite{goldfeld2016semantic}, $j=1,\ldots ,d$, as
    \begin{align}
        \mathcal{T}_{\epsilon_j} &\coloneq \left\{(\boldsymbol{x}^n,\boldsymbol{z}_j^n):\frac{1}{n}i[\boldsymbol{x}^n,\boldsymbol{z}_j^n] \leq I[\boldsymbol{X};\boldsymbol{Z}_j] + \epsilon_j \right\} \label{eq. typical sets},\\
        \mathcal{T}_{\epsilon_j}^c  &\coloneq \left\{(\boldsymbol{x}^n,\boldsymbol{z}_j^n):\frac{1}{n}i[\boldsymbol{x}^n,\boldsymbol{z}_j^n] > I[\boldsymbol{X};\boldsymbol{Z}_j] + \epsilon_j \right\}.
    \end{align}

For simpler notation, we define the two corresponding sets:
    \begin{align}
        \mathcal{T}_{\epsilon_j}(\boldsymbol{x}^n(m,w)) &= \left\{\boldsymbol{z}_j^n \in \mathcal{Z}^{2\times n}: (\boldsymbol{x}^n(m,w),\boldsymbol{z}_j^n) \in \mathcal{T}_{\epsilon_j}\right\},\\
        \mathcal{T}_{\epsilon_j}^c(\boldsymbol{x}^n(m,w)) &= \left\{\boldsymbol{z}_j^n \in \mathcal{Z}^{2\times n}: (\boldsymbol{x}^n(m,w),\boldsymbol{z}_j^n) \in \mathcal{T}_{\epsilon_j}^c\right\}.
    \end{align}

Then the output distribution $P_{\boldsymbol{Z}_j^n|\boldsymbol{X}^n(m,\cdot)}$ can be split into a typical part and an atypical part:
     \begin{equation} \label{eq. split}
         P_{\boldsymbol{Z}_j^n|\boldsymbol{X}^n(m,\cdot)} = P_{j,\mathcal{C}_n,1} +  P_{j,\mathcal{C}_n,2},
     \end{equation} where
     \begin{align}
         P_{j,\mathcal{C}_n,1} &=  \frac{1}{L_1} \sum_{w=1}^{L_1}
         K^{(\boldsymbol{\theta}){\otimes n}}_{{AE}_j}(\cdot \cap \mathcal{T}_{\epsilon_j}(\boldsymbol{x}^n(m,w)) ,\boldsymbol{x}^n(m,w)),\label{eq. P_jC1}\\
         P_{j,\mathcal{C}_n,2} &=  \frac{1}{L_1} \sum_{w=1}^{L_1} 
         K^{(\boldsymbol{\theta}){\otimes n}}_{{AE}_j}(\cdot \cap \mathcal{T}_{\epsilon_j}^c(\boldsymbol{x}^n(m,w)),\boldsymbol{x}^n(m,w)),
     \end{align}
which leads to the following two lemmas.

\begin{lemma}[Bound for Typical Terms]\label{Lem. Bound for Typical Terms}
    For some $\beta_{j1} >0$ where $j=1,\dots,d$:
    \begin{align}  \mathbb{E}_{{\mathcal{C}}_n}\mathbb{E}_{Q_{\boldsymbol{Z}_j}^n}\left[\frac{dP_{j,\mathcal{C}_n,1}}{dQ_{\boldsymbol{Z}_j}^n}-1\right]^+ \leq \exp{\{-n\beta_{j1}\}}.
    \end{align}
\end{lemma}
The proof of Lemma \ref{Lem. Bound for Typical Terms} can be found in Appendix \ref{app:typ-resolv}. We provide a necessary definition before Lemma \ref{Lem. Bound for Atypical Terms}. 

\begin{mydef}[R\'enyi Divergence\cite{van2014renyi}]\label{def. Renyi}
    Let $(\Omega,\Sigma)$ be a measurable space and let $P,Q$ be probability measures on $\Sigma$ with $Q \ll P$. Then for $\alpha$, i.e. $\alpha \in (0,1) \bigcup (1,\infty)$, the R\'enyi Divergence is defined by
    \begin{equation}
        D_{\alpha}[P||Q] = \frac{1}{\alpha-1}\log{\int\left(\frac{dQ}{dP}\right)^{1-\alpha}dP}. \label{eq. def. Renyi}
    \end{equation}
\end{mydef}

\begin{lemma}[Bound for Atypical Terms]\label{Lem. Bound for Atypical Terms}
    For some $\alpha>1$ and $0<\beta_{j2}\leq(\alpha-1)(I[\boldsymbol{X};\boldsymbol{Z}_j]+\epsilon_j-D_{\alpha}[Q_{\boldsymbol{X},\boldsymbol{Z}_j}||Q_{\boldsymbol{X}}Q_{\boldsymbol{Z}_j}])$ where $j=1,\dots,d$:
    \begin{equation}
        Q_{\boldsymbol{X},\boldsymbol{Z}_j}^n(\mathcal{T}_{\epsilon_j}^c) \leq \exp{\{-n\beta_{j2}\}}.
    \end{equation}
\end{lemma}
The proof of Lemma \ref{Lem. Bound for Atypical Terms} can be found in Appendix \ref{app:typ-resolv}.
\begin{proof}[Proof of Lemma \ref{Bound for Expectation of TV Distance}]
     The expectation of TV distance \eqref{f(C)} can be split into two corresponding parts due to the typical and atypical sets:
     \begin{align}              &\mathbb{E}_{\mathcal{C}_n}||P_{\boldsymbol{Z}_j^n|\boldsymbol{X}^n(m,\cdot)} - Q_{\boldsymbol{Z}_j}^n||_{\rm{TV}} \\
         &= \mathbb{E}_{\mathcal{C}_n}\mathbb{E}_{Q_{\boldsymbol{Z}_j}^n}\left[\frac{dP_{j,\mathcal{C}_n,1}}{dQ_{\boldsymbol{Z}_j}^n}+\frac{dP_{j,\mathcal{C}_n,2}}{dQ_{\boldsymbol{Z_j}}^n}-1\right]^+ \label{26}\\
         &\leq  \mathbb{E}_{\mathcal{C}_n}\mathbb{E}_{Q_{\boldsymbol{Z}_j}^n}\left[\frac{dP_{j,\mathcal{C}_n,1}}{dQ_{\boldsymbol{Z}_j}^n}-1\right]^+ +\mathbb{E}_{\mathcal{C}_n}[P_{j,\mathcal{C}_n,2}] \\
         &= \mathbb{E}_{\mathcal{C}_n}\mathbb{E}_{Q_{\boldsymbol{Z}_j}^n}\left[\frac{dP_{j,\mathcal{C}_n,1}}{dQ_{\boldsymbol{Z}_j}^n}-1\right]^+ + Q_{\boldsymbol{X},\boldsymbol{Z_j}}^n(\mathcal{T}_{\epsilon_j}^c),\label{28}
     \end{align}
     where \eqref{26} follows from \eqref{eq. split} and Lemma \ref{rewrite TV distance}. In \eqref{28} $\mathbb{E}_{\mathcal{C}_n}[P_{j,\mathcal{C}_n,2}] = Q_{\boldsymbol{X},\boldsymbol{Z_j}}^n(\mathcal{T}_{\epsilon_i}^c)$ follows since the codewords $\boldsymbol{X}^n(m,w)$ are drawn according to the distribution $Q^n_{\boldsymbol{X}}$.
     
     The terms in \eqref{28} are bounded by Lemma \ref{Lem. Bound for Typical Terms} and Lemma \ref{Lem. Bound for Atypical Terms}, respectively. Combining these two bounds yields a bound for the expectation of the total variation distance:
     \begin{align}          \mathbb{E}_{\mathcal{C}_n}||P_{\boldsymbol{Z}_j^n|\boldsymbol{X}^n} - Q_{\boldsymbol{Z}_j}^n||_{\rm{TV}} &\leq \exp{\{-n\beta_{j1}\}} + \exp{\{-n\beta_{j2}\}} \\
         &\leq \exp{\{-n\beta_{j}\}} 
     \end{align} for any choice of $\beta_{j}$ with $0 < \beta_{j} < \beta_{j1},\beta_{j2}$ and all sufficiently large $n$.
\end{proof}

By Lemma \ref{Bound for Expectation of TV Distance} and \eqref{eq. c bounded}, the McDiarmid inequality \eqref{eq. derived-McDiarmid} yields
\begin{align}
     P_{\mathcal{C}_n}(f_j(\mathcal{C}_n) \!\geq \!\exp\{-n\beta_{j}\}\!+\!\epsilon_j) &\leq \exp\{{-2\epsilon_j^2 L_1}\} \\
     &= \exp\{-2\epsilon_j^2\exp\{nR_1\}\}\label{eq. specify epsilon}.
\end{align}

We can set $\epsilon_j$ to one possible value as follows: 
\begin{equation}
    \epsilon_j \coloneq \frac{1}{\sqrt{2}}\exp{\{-nb_j/2\}},
\end{equation} where $b_j$ is constant for each $j = 1,\dots,d$. The choice of $\epsilon_j$ ensures that the exponential term on the right-hand side in \eqref{eq. specify epsilon} becomes $\exp{\{-\exp\{n(R_1-b_j)\}\}}$ and is double exponentially small when $0<b_j<R_1$ is specified. Therefore by choosing a suitable $\zeta_{j1}>\beta_j$ and setting $\zeta_{j2}\coloneq R_1-b_j$, the inequality can be further simplified to
\begin{equation}
     P_{\mathcal{C}_n}(f_j(\mathcal{C}_n)\geq \exp\{-n\zeta_{j1}\}) \leq \exp\{-\exp\{n\zeta_{j2}\}\}\label{eq. distinguishing security result},
\end{equation} which implies that the probability of not achieving distinguishing security via random codebooks is double exponentially low.

\subsubsection{Reliability}
In the proof, we first define a new type of average decoding error probability. To send one message $m\in \mathcal{M}$, the encoder chooses the randomness $w\in \mathcal{W}$ uniformly at random to generate a codeword $\boldsymbol{x}^n(m,w)$ and transmits it. The decoder reconstructs the message and randomness from the received sequence. The decoding is said to be successful if the reconstructed message and randomness exactly match the sent ones. Otherwise, we denote it as an error event by $\mathcal{E}'$. The new error probability is  given by
    \begin{align}        
    P_{\mathcal{E}'}(\mathcal{C}_n,D_n) = &\frac{1}{LL_1} \sum_{w=1}^{L_1} \sum_{m=1}^L  \label{eq. new average decoding error}\\&P_{{Y}^n|{X}^n}(D({Y}^n) \neq (m,w)|{X}^n={x}^n(m,w))\nonumber.
    \end{align}
    
From Definition \ref{def. error} and \eqref{eq. new average decoding error}, it can be observed that any error in reconstructing $m$ is included in the set of errors in reconstructing $(m,w)$, implying that the error probability of reconstructing 
$m$ is less than or equal to the error probability of reconstructing $(m,w)$, i.e.,

\begin{equation}
    P_{\mathcal{E}}(\mathcal{C}_n,D_n) \leq P_{\mathcal{E}'}(\mathcal{C}_n,D_n). \label{eq. error relation}
\end{equation}

Then the proof of reliability aims to derive an upper bound on the expectation of $P_{\mathcal{E}'}(\mathcal{C}_n,D_n)$ by the following Lemma and apply Markov's inequality to show that $P_{\mathcal{E}'}(\mathcal{C}_n,D_n)$ is low with high probability.

\begin{lemma}[Bound for Expectation of Average Decoding Error]\label{Bound for Expectation of Average Decoding Error}
    Suppose that the moment generating functions of $i[\boldsymbol{X};\boldsymbol{Y}]$ exist and are finite for some $t > 0$, and the randomness rate satisfies $R_1 < I[\boldsymbol{X};\boldsymbol{Y}]-R_s$, then there exist  $\Gamma > 0$ such that
    \begin{equation}
    \mathbb E_{\mathcal{C}_n}[P_{\mathcal{E}'}(\mathcal{C}_n,D_n)] \leq \exp\{-n\Gamma\},\label{eq. bound for expectation of error}
    \end{equation} 
    for all sufficiently large $n\in\mathbb{N}$.
\end{lemma}

Before proving Lemma \ref{Bound for Expectation of Average Decoding Error}, we split the decoding error by adopting the method from \cite{frey2017mac} and state an auxiliary result for each part.

We define the  typical set and atypical set in a similar way to the proof of distinguishing security:
    \begin{align}
        \mathcal{T}^{'}_{\epsilon} &\coloneq \left\{(\boldsymbol{x}^n,\boldsymbol{y}^n):\frac{1}{n}i[\boldsymbol{x}^n;\boldsymbol{y}^n] \leq I[\boldsymbol{X};\boldsymbol{Y}] + \epsilon \right\},\label{eq. def. typical set reliability}\\
        \mathcal{T}_{\epsilon}^{'c}  &\coloneq \left\{(\boldsymbol{x}^n,\boldsymbol{y}^n):\frac{1}{n}i[\boldsymbol{x}^n;\boldsymbol{y}^n] > I[\boldsymbol{X};\boldsymbol{Y}] + \epsilon \right\},
    \end{align}
and a corresponding joint typicality decoder:
\begin{align*}
    D(\boldsymbol{y}^n) \!\coloneq \!
    \begin{cases}
        (m,w), \! &(\boldsymbol{x}^n(m,w),\boldsymbol{y}^n) \! \in \! \mathcal{T}^{'}_{\epsilon} \ \text{for \! unique} \ (m,w),\\
        (1,1),\! &\text{otherwise.}
    \end{cases}      
\end{align*} 

For simpler notation, we define the two corresponding sets:
\begin{align}
    \mathcal{T}^{'}_{\epsilon}(\boldsymbol{x}^n) &= \{\boldsymbol{y}^n:(\boldsymbol{x}^n,\boldsymbol{y}^n)\in \mathcal{T}^{'}_{\epsilon}\},\\
    \mathcal{T}^{'c}_{\epsilon}(\boldsymbol{x}^n) &= \{\boldsymbol{y}^n:(\boldsymbol{x}^n,\boldsymbol{y}^n)\in \mathcal{T}^{'c}_{\epsilon}\}.  
\end{align}

In order to achieve a successful decoding, the received sequence must be jointly typical with the unique encoded message and unique with this property. The error event $\mathcal{E}'$ can be decomposed into two distinct events:  $\mathcal{E}_1$ and $\mathcal{E}_2$. Event $\mathcal{E}_1$ occurs when the received sequence $\boldsymbol{y}^n$ can not be assigned to any jointly typical codeword, while event $\mathcal{E}_2$ occurs when $\boldsymbol{y}^n$ is assignable to multiple jointly typical codewords. The probabilities of the error events are formalized as follows:
    \begin{align}
    P_{\mathcal{E}_1}(m,w)&\coloneq \!P_{\boldsymbol{Y}^n|\boldsymbol{X}^n}(\boldsymbol{Y}^n \!\in \!\mathcal{T}^{'c}_{\epsilon}(\boldsymbol{X}^n) |\boldsymbol{X}^n \!= \boldsymbol{x}^n(m,w)),\label{eq. P_epsilon1}\\
     P_{\mathcal{E}_2}(m,w)&\coloneq\! P_{\boldsymbol{Y}^n|\boldsymbol{X}^n}(\boldsymbol{Y}^n \in \mathcal{T}^{'}_{\epsilon}(\boldsymbol{X}^n) |\boldsymbol{X}^n \!= \boldsymbol{x}^n(m',w'),\nonumber\\
     & \qquad \qquad \ \; (m',w')\neq (m,w)).\label{eq. P_epsilon2}
    \end{align}

\begin{lemma}\label{lem. bound for error one}
    For some $\alpha>1$ and $0<\Gamma_1\leq(\alpha-1)(I[\boldsymbol{X};\boldsymbol{Y}]+\epsilon-D_{\alpha}[Q_{\boldsymbol{X},\boldsymbol{Y}}||Q_{\boldsymbol{X}}Q_{\boldsymbol{Y}}])$:
    \begin{equation}
        \mathbb E_{\mathcal{C}_n}[P_{\mathcal{E}_1}(m,w)] \leq \exp\{-n\Gamma_1\}.
    \end{equation}   
\end{lemma}

\begin{lemma}\label{lem. bound for error two}
    For some $0<\Gamma_2<I[\boldsymbol{X};\boldsymbol{Y}]+\epsilon-R-R_1$:
    \begin{equation}
        \mathbb E_{\mathcal{C}_n}[P_{\mathcal{E}_2}(m,w)]\leq \exp\{-n\Gamma_2\}.
    \end{equation}
\end{lemma}
The proofs of Lemma \ref{lem. bound for error one} and Lemma \ref{lem. bound for error two} can be found in Appendix \ref{app:errors}.

\begin{proof}[Proof of Lemma \ref{Bound for Expectation of Average Decoding Error}]
Due to the splitting \eqref{eq. P_epsilon1} and \eqref{eq. P_epsilon2}, the expectation of the average decoding error probability \eqref{eq. average decoding error} can be bounded by
    \begin{align}
    &\mathbb E_{\mathcal{C}_n}[P_{\mathcal{E}'}(\mathcal{C}_n,D_n)]\nonumber\\ 
    \leq& \mathbb E_{\mathcal{C}_n}\left[\frac{1}{LL_1} \sum_{w=1}^{L_1} \sum_{m=1}^L (P_{\mathcal{E}_1}(m,w) + P_{\mathcal{E}_2}(m,w))\right]\label{eq. union bound}\\
    =&\frac{1}{LL_1} \sum_{w=1}^{L_1} \sum_{m=1}^L \mathbb E_{\mathcal{C}_n}\left[P_{\mathcal{E}_1}(m,w)\right] 
    \nonumber\\ &  +\frac{1}{LL_1} \sum_{w=1}^{L_1} \sum_{m=1}^L \mathbb E_{\mathcal{C}_n}\left[P_{\mathcal{E}_2}(m,w)\right]\label{eq. linearity of expectation}\\
    =& \mathbb E_{\mathcal{C}_n}[P_{\mathcal{E}_1}(m,w)] +\mathbb E_{\mathcal{C}_n}[P_{\mathcal{E}_2}(m,w)]\label{eq.i.i.d. codewords},
    \end{align} 
    where \eqref{eq. union bound} is an application of the union bound, \eqref{eq. linearity of expectation} follows by the linearity of expectation, and \eqref{eq.i.i.d. codewords} follows by the expectation independent from the individual messages due to the i.i.d. codewords. The two terms in \eqref{eq.i.i.d. codewords} are respectively bounded by Lemma \ref{lem. bound for error one} and Lemma \ref{lem. bound for error two}, yielding 
\begin{align}
    \mathbb E_{\mathcal{C}_n}[P_{\mathcal{E}'}(\mathcal{C}_n,D_n)]
    &\leq \exp\{-n\Gamma_1\}+\exp\{-n\Gamma_2\}\nonumber\\   
    &\leq \exp\{-n\Gamma\}\label{eq. bound for expectation of error} 
\end{align} for any choice of $\Gamma$ with $0<\Gamma<\Gamma_1,\Gamma_2$ and for all sufficiently large $n$.
\end{proof}

Next we can apply Markov's inequality with \eqref{eq. bound for expectation of error} for $\Gamma>\gamma_1>0$:
\begin{align}
    P_{\mathcal{C}_n}&(P_{\mathcal{E}'}(\mathcal{C}_n,D_n) 
    \geq \exp\{-n\gamma_1\})\\&\leq \exp\{n\gamma_1\}\mathbb E_{\mathcal{C}_n}[P_{\mathcal{E}'}(\mathcal{C}_n,D_n)]\\
    &\leq \exp\{n\gamma_1\}\exp\{-n\Gamma\},\\
    &=\exp\{-n(\Gamma-\gamma_1)\},\\
    &=\exp\{-n\gamma_2\}
\end{align} with the choice $\gamma_2 \coloneq \Gamma-\gamma_1$.

Hence by \eqref{eq. error relation}, we have
\begin{equation}
    P_{\mathcal{C}_n}\left(P_{\mathcal{E}}(\mathcal{C}_n,D_n)  \geq \exp\{-n\gamma_1\}\right)\leq \exp\{-n\gamma_2\}.\label{eq. reliability result}
\end{equation}

\subsubsection{Distinguishing Security Achievability}
Now we can combine the distinguishing security result \eqref{eq. distinguishing security result} and reliability result \eqref{eq. reliability result} into a single statement of achievability:

For $\max\limits_{j=1,\dots,d}I[\boldsymbol{X};\boldsymbol{Z}_j] < R_1 <I[\boldsymbol{X};\boldsymbol{Y}]-R_s$ and $j=1,\dots,d$, we have
   \begin{align*}
      \bullet P_{\mathcal{C}_n}&\left(||P_{\boldsymbol{Z}_j^n|\boldsymbol{X}^n(m,\cdot)} - Q_{\boldsymbol{Z}_j}^n||_{\rm{TV}}\geq \exp\{-n\zeta_{j1}\}\right)\nonumber\\ &\leq \exp\{-\exp\{n\zeta_{j2}\}\},\\ 
    \bullet P_{\mathcal{C}_n}&\left(P_{\mathcal{E}}(\mathcal{C}_n,D_n) \right) \geq \exp\{-n\gamma_1\})\leq \exp\{-n\gamma_2\}.
\end{align*}

Therefore via the union bound, we have the probability of not achieving reliable communication for Bob and secure communication for all Eves:
\begin{align}
    &P_{\mathcal{C}_n}(\bigcup_{j=1,\dots,d}||P_{\boldsymbol{Z}_j^n|\boldsymbol{X}^n(m,\cdot)} - Q_{\boldsymbol{Z}_j}^n||_{\rm{TV}}\geq \exp\{-n\zeta_{j1}\} \nonumber\\
    &\qquad \bigcup P_{\mathcal{E}'}(\mathcal{C}_n,D_n)\geq \exp\{-n\gamma_1\} )\\
    &\leq \sum\limits_{j=1,\dots,d}P_{\mathcal{C}_n}(||P_{\boldsymbol{Z}_j^n|\boldsymbol{X}^n(m,\cdot)} - Q_{\boldsymbol{Z}_j}^n||_{\rm{TV}}\geq \exp\{-n\zeta_{j1}\})\nonumber \\
    &\qquad+ P_{\mathcal{C}_n}(P_{\mathcal{E}'}(\mathcal{C}_n,D_n) \geq \exp\{-n\gamma_1\})\\
    &\leq \sum\limits_{j=1,\dots,d}\exp\{-\exp\{n\zeta_{j2}\}\} + \exp\{-n\gamma_2\}\\
    &\leq \exp\{-n\lambda_1\}
\end{align} for any choice of $\lambda_1$ with $0<\lambda_1 < \gamma_2, \min\limits_{j=1,\dots,d}(\zeta_{j2})$ and sufficiently large $n$.

\section{Evaluation of Secrecy Rate} \label{Sec. Evaluation}
In this section, we evaluate the achievable secrecy rate of the RIS-assisted wiretap channel with SSOC. Then we formulate the secrecy rate maximization problem and design an optimization algorithm. In Section \ref{Sec. Evaluation} and \ref{Sec. Simulation}, we assume for numerical simplicity that all Eves are able to extract the corresponding signals from the direct link and the RIS-assisted link, and that the two signals are mutually independent. The theoretical result of Theorem \ref{Distinguishing achievablity} still holds for Eves' channels as described in \eqref{eq.origRecEve}

\subsection{Secrecy Rate Maximization Problem}
Consider a RIS-assisted wiretap channel $T_{\text{RIS}}$ with an AWGN model as defined in Sec. \ref{Subsec. Wireless}. Let $\boldsymbol{X}^n \in \mathbb{C}^{2\times n}$ be the vector-valued channel input, $\boldsymbol{Y}^n \in \mathbb{C}^{2\times n}$ and $\boldsymbol{Z}_j^n \in \mathbb{C}^{2\times n}$ be the vector-valued channel outputs for Bob and Eve$_j$, $j=1,\dots,d$, respectively. The channel outputs can be written as
\begin{align} 
    \boldsymbol{Y}^n &= \bH_{AB}^{(\boldsymbol{\theta})} \boldsymbol{X}^n  + \boldsymbol{E}^n_{AB},\\
    \boldsymbol{Z}_j^n &= \bH_{{AE}_j}^{(\boldsymbol{\theta})} \boldsymbol{X}^n  + \boldsymbol{E}^n_{AE_j},\label{eq. Output Z}\\
    \bH_{AB}^{(\boldsymbol{\theta})} &= \left[
    \begin{array}{cc}
    \alpha_{AB,1}& 0	\\
    0 & \alpha_{AB,2}^{(\boldsymbol{\theta})}
    \end{array}
    \right], \\
   \bH_{{AE}_j}^{(\boldsymbol{\theta})} &= \left[
    \begin{array}{cc}
    \alpha_{{AE}_j,1}& 0	\\
    0 & \alpha_{{AE}_j,2}^{(\boldsymbol{\theta})}
    \end{array}
    \right],
\end{align}
where $\bH_{AB}^{(\boldsymbol{\theta})}\in \mathbb{C}^{2\times 2}$ denotes the channel matrix between Alice and Bob, and the diagonal corresponds to the independence between direct link and RIS-assisted link in \eqref{eq.combineRec2}. $\boldsymbol{E}^n_{AB} = [E^n_{AB,1},E^n_{AB,2}]^T\in \mathbb{C}^{2\times n}$ denotes the channel noise between Alice and Bob, where $E^n_{AB,1},E^n_{AB,2}$ correspond to the independent complex AWGNs from the two links with noise power $N_{AB,1}$, $N_{AB,2}$ in \eqref{eq.combineRec2}, respectively. According to the assumption, the channel output at Eve$_j$ can be similarly simplified to \eqref{eq. Output Z} from \eqref{eq.origRecEve}. The channel matrix $\mathbf{H}_{{AE}_j}^{(\boldsymbol{\theta})} \in \mathbb{C}^{2\times 2}$ between Alice and Eve$_j$ is also a diagonal matrix with $\alpha_{{AE}_j,1}$, $\alpha_{{AE}_j,2}^{(\boldsymbol{\theta})}$ on the diagonal, and $\boldsymbol{E}^n_{AE_j} \!=\![E^n_{AE_j,1},E^n_{AE_j,2}]^T \!\in\! \mathbb{C}^{2\times n}$, where the components are also independent complex AWGNs with noise power $N_{AE_j,1}, N_{AE_j,2}$.
The RIS-related channel gains are optimized and fixed by the RIS configuration. The power is constrained according to $\mathbb{E}(|{X_1^n}^T X_1^n|^2) \leq P_1$ and $\mathbb{E}(|{X_2^n}^T X_2^n|^2) \leq P_2$. The total transmit power is constrained by $P_1+P_2 \leq P_t$.

Now we can evaluate the achievable secrecy rate from Theorem \ref{Distinguishing achievablity}. Assume that $\boldsymbol{X}^n$ is a Gaussian vector with zero mean and covariance matrix $K_{\boldsymbol{X}^n}$ subject to the power constraint $\rm{Tr}$$(K_{\boldsymbol{X}^n})=P_1+P_2 \leq P_t$, where $\rm Tr(\cdot)$ denotes the sum of elements on the main diagonal of the matrix. The secrecy rate \cite[Example 22.3]{el2011network} is given by
\begin{align} \label{eq. Rs det}
    R_s &= \log\det(\mathbb{I}+\bH_{AB}^{(\boldsymbol{\theta})}K_{\boldsymbol{X}^n}\bH_{AB}^{(\boldsymbol{\theta})H}K^{-1}_{\boldsymbol{E}_{AB}^n}) \nonumber\\
    &-\max\limits_{j=1,\dots,d} \left\{\log\det(\mathbb{I}+\bH_{AE_j}^{(\boldsymbol{\theta})}K_{\boldsymbol{X}^n}\bH_{AE_j}^{(\boldsymbol{\theta})H}K^{-1}_{\boldsymbol{E}_{AE_j}^n})\right\},
\end{align}where $\mathbb{I}$ is an identity matrix of size $2 \times 2$, and $K_{\boldsymbol{E}_{AB}^n},K_{\boldsymbol{E}_{AE_j}^n}$ are the covariance matrices of $\boldsymbol{E}^n_{AB},\boldsymbol{E}^n_{AE_j}$, respectively.

Since the noise components in $\boldsymbol{E}^n_{AB}$ and $\boldsymbol{E}^n_{AE_j}$ are independent, $K_{\boldsymbol{E}_{AB}^n}$ and $K_{\boldsymbol{E}_{AE_j}^n}$ are diagonal matrices with the corresponding noise power on the diagonal. The secrecy rate can be further simplified to
\begin{align}\label{eq. independent noise}
    &R_s(P_1,P_2) = \left.\log\left(1+\frac{P_1|\alpha_{AB,1}|^2}{N_{AB,1}}\right)+\log\left(1+\frac{P_2|\alpha_{AB,2}^{(\boldsymbol{\theta})}|^2}{N_{AB,2}}\right)\right.\\ \nonumber
&-  \max\limits_{j=1,\dots,d} \left(\log\left(1+\frac{P_1|\alpha_{AE_j,1}|^2}{N_{AE_j,1}}\right)+\log\left(1+\frac{P_2|\alpha_{AE_j,2}^{(\boldsymbol{\theta})}|^2}{N_{AE_j,2}}\right)\right).
\end{align}

In order to enhance the security of the system, we jointly optimize $P_1$ and $P_2$ to maximize the secrecy rate for all Eves with fixed $\boldsymbol{\theta}$. The problem of interest can be formulated as
\begin{align}\label{eq. Op Problem}
    &\max\limits_{P_1,P_2\geq 0 } R_s(P_1,P_2) \nonumber\\ 
    &\text{subject to}  \quad P_1+P_2\leq P_t. 
\end{align}

\subsection{Algorithm Design for Problem \eqref{eq. Op Problem}}
Observing that the objective function in \eqref{eq. independent noise} is non-concave, we adopt the Minorize-Maximization (MM) method \cite{hunter2004tutorial} to tackle the problem. The idea is to minorize the objective function by one suitable surrogate function in each iteration, making it concave. Then we solve the concave optimization problem for the surrogate function instead of the objective function with the Karush-Kuhn-Tucker (KKT) approach \cite{kuhn2013nonlinear} and produce the pair of $(P_1,P_2)$ for the next iteration until convergence. 

We start to explain one iteration of the proposed algorithm. Let $\mu_1 \coloneq \frac{|\alpha_{AB,1}|^2}{N_{AB,1}}, \mu_2 \coloneq \frac{|\alpha_{AB,2}^{(\boldsymbol{\theta})}|^2}{N_{AB,2}}, \beta_{j,1} \coloneq\frac{|\alpha_{AE_j,1}|^2}{N_{AE_j,1}}, \beta_{j,2} \coloneq\frac{|\alpha_{AE_j,2}^{(\boldsymbol{\theta})}|^2}{N_{AE_j,2}}$, $j=1,\dots,d,$ to simplify the notations.

Note that $\log(1+P_1\beta_{j,1})$ and $\log(1+P_2\beta_{j,2})$ are two logarithmic functions, their tangent lines at the fixed point $(P^{(m)}_1,P^{(m)}_2)$ are given by
\begin{align}
    T_{\beta_{j,1}}\!(P_1|P^{(m)}_1\!) &\!=\! \log(1+\beta_{j,1}P^{(m)}_1\!)+\frac{\beta_{j,1}(P_1-P^{(m)}_1)}{1+\beta_{j,1}P^{(m)}_1},\\
    T_{\beta_{j,2}}\!(P_2|P^{(m)}_2\!) &\!=\! \log(1+\beta_{j,2}P^{(m)}_2\!)+\frac{\beta_{j,2}(P_2-P^{(m)}_2)}{1+\beta_{j,2}P^{(m)}_2},
\end{align} where $(P^{(m)}_1,P^{(m)}_2)$ denotes a pair of fixed value of $(P_1,P_2)$ in the iteration $m$.

Due to the concavity of logarithmic functions, we have for all $(P_1,P_2)$
\begin{align}
    T_{\beta_{j,1}}(P_1|P^{(m)}_1) &\geq \log(1+P_1\beta_{j,1}),\\
    T_{\beta_{j,2}}(P_2|P^{(m)}_2) &\geq \log(1+P_2\beta_{j,2}).
\end{align}

We can construct a surrogate function for $R_s(P_1,P_2)$ by
\begin{align} \label{eq. surrogate function}
    &g(P_1,P_2,m|P^{(m)}_1,P^{(m)}_2) \nonumber\\&\coloneq \left.\log(1+P_1\mu_1)+\log(1+P_2\mu_2)\right. \nonumber\\
    &-\max\limits_{j=1,\dots,d} \left(T_{\beta_{j,1}}(P_1|P^{(m)}_1)+T_{\beta_{j,2}}(P_2|P^{(m)}_2)\right),
\end{align}

The function $g(P_1,P_2,m|P^{(m)}_1,P^{(m)}_2)$ is said to minorize $R_s(P_1,P_2)$ at the point $(P^{(m)}_1,P^{(m)}_2)$ \cite[Equation (1)]{hunter2004tutorial} if for all $(P_1,P_2)$
\begin{align}
    g(P_1,P_2,m|P^{(m)}_1,P^{(m)}_2) &\leq R_s(P_1,P_2),\label{eq. surrogate1}\\
    g(P^{(m)}_1,P^{(m)}_2,m|P^{(m)}_1,P^{(m)}_2) &= R_s(P^{(m)}_1,P^{(m)}_2).\label{eq. surrogate2}
\end{align}

In the MM algorithm, we maximize the surrogate function $g(P_1,P_2,m|P^{(m)}_1,P^{(m)}_2)$ rather than the actual function $R_s(P_1,P_2)$. If $(P^{(m+1)}_1,P^{(m+1)}_2)$ denotes the maximizer of $g(P_1,P_2,m|P^{(m)}_1,P^{(m)}_2)$, i.e., 
\begin{align}
&g(P^{(m+1)}_1,P^{(m+1)}_2,j|P^{(m)}_1,P^{(m)}_2) \nonumber\\ \geq \quad&  g(P_1,P_2,m|P^{(m)}_1,P^{(m)}_2) \quad \text{for all}\;(P_1,P_2),\label{eq. maximizer}
\end{align} 
then we can show that the MM procedure pushes $R_s(P_1,P_2)$  towards its maximum value: 
\begin{align}
    R_s(P^{(m+1)}_1,P^{(m+1)}_2) &\geq g(P^{(m+1)}_1,P^{(m+1)}_2,m|P^{(m)}_1,P^{(m)}_2)\nonumber\\
    &\geq g(P^{(m)}_1,P^{(m)}_2,m|P^{(m)}_1,P^{(m)}_2)\nonumber\\
    &= R_s(P^{(m)}_1,P^{(m)}_2)， \label{eq. ascent}
\end{align} where the inequality follows directly from \eqref{eq. surrogate1}, \eqref{eq. surrogate2} and \eqref{eq. maximizer}. The ascent property \eqref{eq. ascent} provides the MM algorithm with remarkable numerical stability.

In the next step, we perform optimization for the surrogate function over $P_1,P_2$. The corresponding problem is
\begin{align}
    &\max\limits_{P_1,P_2\geq 0} \ g(P_1,P_2,m|P^{(m)}_1,P^{(m)}_2)\nonumber \\
    &\text{subject to}  \quad P_1+P_2\leq P_t. \label{eq. op KKT}
\end{align} 

Note that $T_{\beta_{j,1}}(P_1|P^{(m)}_1),T_{\beta_{j,2}}(P_2|P^{(m)}_2)$ are two linear functions. Taking the maximum of  $T_{\beta_{j,1}}(P_1|P^{(m)}_1)+T_{\beta_{j,2}}(P_2|P^{(m)}_2)$ yields a convex function. The minus sign ahead makes the max function concave. Since logarithmic functions are concave, $g(P_1,P_2,m|P^{(m)}_1,P^{(m)}_2)$ is a sum of concave functions, and hence it is also a concave function. Therefore we can adopt the KKT algorithm to solve the optimization problem \eqref{eq. op KKT}. 

We first form a Lagrangian function with a Lagrange multiplier $\lambda$:
\begin{equation} \label{eq. Lagrangian}
    L(P_1,P_2,\lambda) \coloneq g(P_1,P_2,m|P^{(m)}_1,P^{(m)}_2)-\lambda q(P_1,P_2),
\end{equation}where 
\begin{equation}\label{eq. constraint}
    q(P_1,P_2) \coloneq P_1+P_2-P_t \leq 0
\end{equation} represents the inequality constraint.

Since $g(P_1,P_2,m|P^{(m)}_1,P^{(m)}_2)$ is a concave function and the constraint \eqref{eq. constraint} clearly shows that the optimization problem \eqref{eq. op KKT} is strictly feasible (Slater condition) \cite[Definition 10.5]{lauritzen2013undergraduate}, the conditions \eqref{eq. P1}-\eqref{eq. check2} are necessary and sufficient for $(P^*_1,P^*_2)$ being the maximizer of \eqref{eq. op KKT} \cite[Theorem 10.6]{lauritzen2013undergraduate}.

\begin{multicols}{2}
\noindent
    \begin{align*}
    \frac{\partial L(P^*_1,P^*_2,\lambda^*)}{\partial P_1} &= 0 \tag{a}\label{eq. P1}\\
    \lambda^* q(P^*_1,P^*_2) &= 0  \tag{c} \label{eq. g=0}\\
    P^*_1, P^*_2, \lambda^* &\geq 0 \tag{e}\label{eq. check2}
    \end{align*}
    \begin{align*}
    \frac{\partial L(P^*_1,P^*_2,\lambda^*)}{\partial P_2} &= 0 \tag{b}\label{eq. P2}\\
    q(P^*_1,P^*_2) &\leq 0  \tag{d}\label{eq. check1} 
    \end{align*}
\end{multicols}

The calculation of the closed-form solution is executed sequentially in two cases:

\emph{Case \Romannum{1}: inactive inequality constraint ($\lambda^* = 0$)}
\vspace*{-0.2cm}
\begin{itemize}
    \item Solving $P^*_1$ and $P^*_2$ from \eqref{eq. P1} and \eqref{eq. P2}.
    \item Putting $P^*_1,P^*_2$ in \eqref{eq. check1} and \eqref{eq. check2} to check the constraints.
    \item If satisfying the constraints, calculating $L(P^*_1,P^*_2, 0)$ as result. Otherwise, going to \emph{case \Romannum{2}}.
\end{itemize}

By applying \eqref{eq. surrogate function}, \eqref{eq. Lagrangian}, \eqref{eq. constraint} in \eqref{eq. P1} and \eqref{eq. P2}, we have
\begin{align*}
        \frac{\mu_1}{(1+P^*_1\mu_1)}-\max\limits_j(\frac{\beta_{j,1}}{1+\beta_{j,1}P^{(m)}_1}) &= 0 , \tag{a'}\\
        \frac{\mu_2}{(1+P^*_2\mu_2)}-\max\limits_j(\frac{\beta_{j,2}}{1+\beta_{j,2}P^{(m)}_2})  &= 0 .\tag{b'}
\end{align*}

In this case, 
\begin{align}
    P^*_1 &= \max\limits_j (\frac{1+\beta_{j,1}P^{(m)}_1}{\beta_{j,1}}) - \frac{1}{\mu_1},\\
    P^*_2 &= \max\limits_j (\frac{1+\beta_{j,2}P^{(m)}_2}{\beta_{j,2}}) - \frac{1}{\mu_2}.
\end{align}

If $P^*_1,P^*_2$ satisfy the inequality \eqref{eq. check1} and are non-negative \eqref{eq. check2}, the maximum is $L(P^*_1,P^*_2,0)$.

\emph{Case \Romannum{2}: active inequality constraint ($\lambda^* > 0$)}
\vspace*{-0.2cm}
\begin{itemize}
    \item $q(P^*_1,P^*_2)=0$ derived from \eqref{eq. g=0}.
    \item Solving $P^*_1,P^*_2,\lambda^*$ from \eqref{eq. P1}, \eqref{eq. P2} and $q(P_1,P_2)=0$.
    \item Putting $P^*_1,P^*_2,\lambda^*$ in \eqref{eq. check2} to check the constraints.
    \item If satisfying the constraints, calculating $L(P^*_1,P^*_2, \lambda^*)$ as result. Otherwise, no results.
\end{itemize}

In this case, there is a system of three linear equations with three unknowns:
\begin{align*}
        \frac{\mu_1}{(1+P^*_1\mu_1)}-\max\limits_j(\frac{\beta_{j,1}}{1+\beta_{j,1}P^{(m)}_1}) -\lambda &= 0 ,  \tag{a'}\label{eq. taga'}\\
        \frac{\mu_2}{(1+P^*_2\mu_2)}-\max\limits_j(\frac{\beta_{j,2}}{1+\beta_{j,2}P^{(m)}_2}) -\lambda &= 0 , \tag{b'}\label{eq. tagb'}\\
        \qquad P^*_1+P^*_2-P_t &= 0  \tag{c'}\label{eq. tagc'}.      
\end{align*} 

After deriving $P^*_2=P_t-P^*_1$ from \eqref{eq. tagc'} and bringing it to \eqref{eq. tagb'}, we can obtain a quadratic equation of $P^*_1$ by eliminating $\lambda^*$ via \eqref{eq. taga'} = \eqref{eq. tagb'}:
\begin{equation}\label{eq. quadractic equation for P1}
    AP^2_1 + BP_1 +C = 0, 
\end{equation}where 
\begin{align*}
    A &\coloneq D\mu_1\mu_2,\\
    B &\coloneq -D(\mu_1-\mu_2)-\mu_1\mu_2(DP_t+2),\\
    C &\coloneq \mu_1-\mu_2-D+P_t(\mu_1\mu_2-D\mu_2),\\
    D &\coloneq \max\limits_j(\frac{\beta_{j,1}}{1+\beta_{j,1}P^{(m)}_1}) - \max\limits_j(\frac{\beta_{j,2}}{1+\beta_{j,2}P^{(m)}_2}).
\end{align*}

If the discriminant $\bigtriangleup = B^2-4AC \geq 0$ for \eqref{eq. quadractic equation for P1}, the roots exist:
\begin{equation}
    P^*_1 = \frac{-B\pm \sqrt{\bigtriangleup}}{2A}.
\end{equation}

Then we have
\begin{align}
    P^*_2 &= P_t - P^*_1,\\
    \lambda^* &= \frac{\mu_1}{(1+ P^*_1\mu_1)}-\max\limits_j(\frac{\beta_{j,1}}{1+\beta_{j,1}P^{(m)}_1}) .
\end{align}

If $P^*_1,P^*_2,\lambda^*$ satisfy the non-negative constraints \eqref{eq. check2}, the maximum is $L(P^*_1,P^*_2, \lambda^*)$.

Eventually, we obtain the maximum value of the surrogate function:
\begin{align}
    \max\limits_{P_1,P_2} \ g(P_1,P_2,m|P^{(m)}_1,P^{(m)}_2) 
    = L(P^*_1,P^*_2, \lambda^*)
\end{align}and produce a new pair for the next iteration:
\begin{equation}
    (P^{(m+1)}_1,P^{(m+1)}_2) = (P^*_1,P^*_2).
\end{equation}

We summarize the whole process of the proposed algorithm for the problem \eqref{eq. Op Problem} in Algorithm \ref{Algorithm 1}. 
\begin{algorithm}
    \caption{Proposed Algorithm for Problem \eqref{eq. Op Problem}}
    \label{Algorithm 1}
    \begin{algorithmic}[1]
        \State Initialization: set an initial pair $(P^{(0)}_1,P^{(0)}_2)$.
        \Repeat 
        \State Construct a surrogate function $g(P_1,P_2,m|P^{(m)}_1\!,P^{(m)}_2\!)$ at the fixed pair $(P^{(m)}_1,P^{(m)}_2)$.
        \State Use KKT approach to find the optimal solution for $g$.
        \State Produce the next iterate pair $(P^{(m+1)}_1,P^{(m+1)}_2)$.
        \Until{convergence.}
    \end{algorithmic}
\end{algorithm}

\begin{Remark}
    The MM algorithm has a linear rate of convergence \cite[Equation (18)]{hunter2004tutorial} as it approaches an optimum point $(P^*_1,P^*_2)$:
    \begin{equation*}
        \lim\limits_{m \rightarrow \infty} \frac{\lVert (P^{(m+1)}_1,P^{(m+1)}_2) - (P^*_1,P^*_2) \rVert}{\lVert (P^{(m)}_1,P^{(m)}_2) - (P^*_1,P^*_2) \rVert} = c <1.
    \end{equation*}
\end{Remark}

\section{Simulation Results} \label{Sec. Simulation}
In this section, we evaluate the performance of the proposed method via simulations for a RIS-assisted system with SSOC. Moreover, we also consider the conventional case without SSOC \cite{shen2019secrecy} as a comparison and the simple case without RIS as a benchmark. We  maximize the secrecy rate by optimizing the power allocation to the direct and RIS-assisted links, while conventional algorithms optimize the total transmit power. Alice, the RIS, and Bob are located at (0,0), (50,0), and (50,10) in meter, respectively, as shown in Fig. \ref{Fig. new}. We set all antenna gains to 5 dB and all noise power to -104 dBm. The path loss model is free space path loss and given by $\rm PL = 20\log_{10} (\frac{4\pi d}{\lambda})$ (dB), where $\lambda = 0.01$ m is the wavelength and $d$ is the distance between two entities. The RIS was configured and fixed during the simulation. The number of iterations in Algorithm \ref{Algorithm 1} is 500.

\begin{figure}[tb]
	\begin{minipage}[t]{0.49\linewidth}
		\centering
		\includegraphics[width=1.1\textwidth]{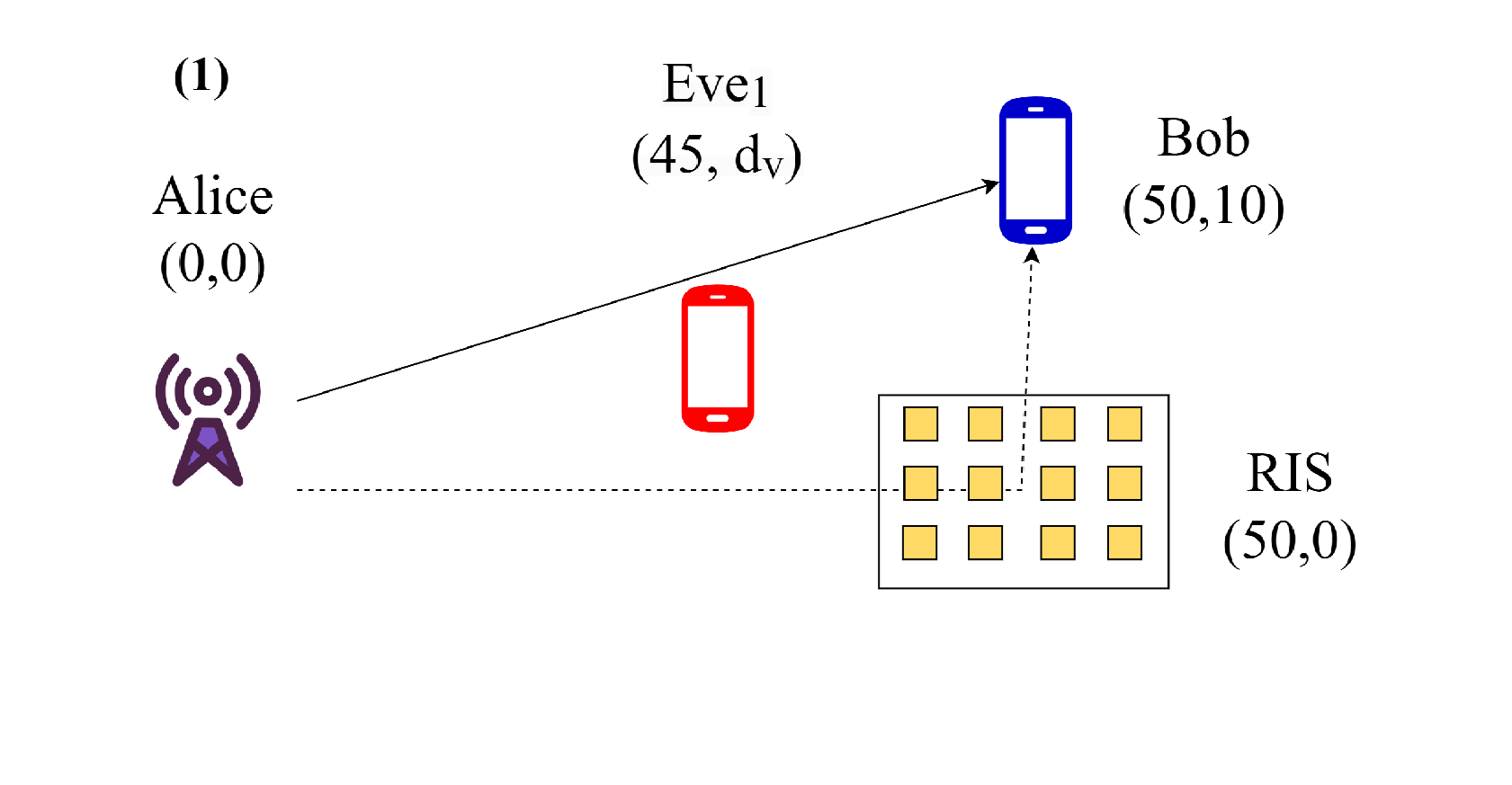}
	\end{minipage}
	\begin{minipage}[t]{0.49\linewidth}
		\centering
		\includegraphics[width=1.1\textwidth]{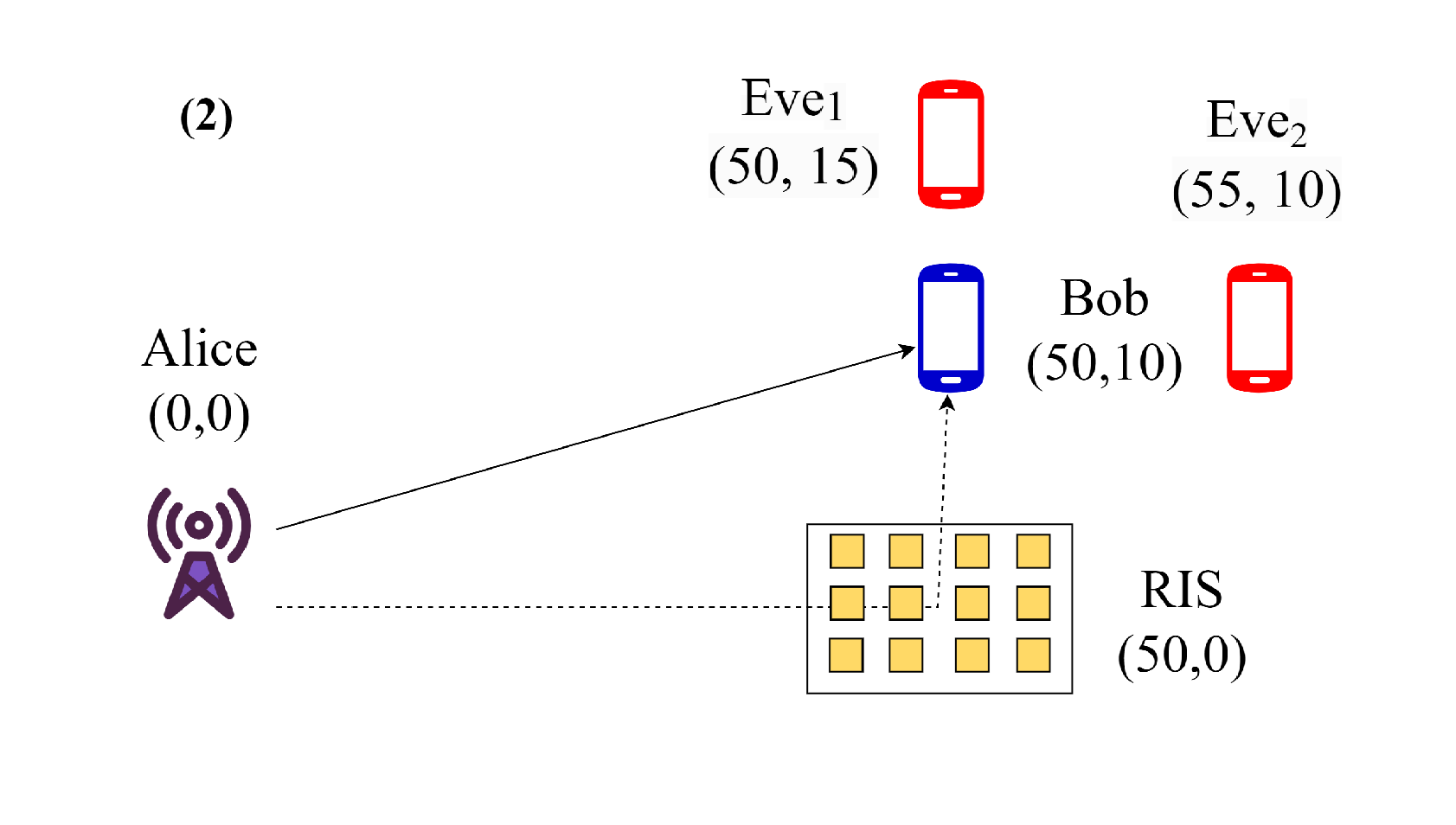}
	\end{minipage}
    	\caption{Simulation settings showing the positions of Alice, the RIS and Bob, along with two scenarios of Eves: (1) one moving Eve (2) two fixed Eves.}
		\label{Fig. new}
\end{figure}

Fig. \ref{Fig.2} shows the secrecy rate performance of different schemes by varying Eve's vertical distance $d_v$. Eve is located at $(45,d_v)$ in meter as shown in Fig. \ref{Fig. new} (1). From the overall perspective, the secrecy rate increases with the growing vertical distance for all three schemes. This is because Eve is moving away from Alice and the RIS, thus reducing the information leakage. It is observed that when $d_v \in [0,5]$, the secrecy rate is zero due to the proximity of Eve to Alice and the RIS. When $d_v \in [5,25]$, the rates of the two schemes with RIS increase significantly and identically, while the rate of the case without RIS stays zero. This is because Eve is still in an advantageous position for intercepting communication over the direct channel, and the RIS-assisted channel takes over the communication task. Therefore, all transmit power is used for the communication over the RIS-assisted channel and the combiner does not help for data extraction from two channels, making the blue and red curves overlap. As $d_v$ further increases, our proposed scheme becomes advantageous over the reference scheme. When both direct and RIS-assisted channels are available, our proposed SSOC method considers their channel state information individually to optimize the power allocation, resulting in a higher secrecy rate. 

Fig. \ref{Fig.3} illustrates the effect of the total transmit power $P_t$ on the secrecy rate. Two Eves are located at (50,15) and (55,10) in meter, respectively, as shown in Fig. \ref{Fig. new} (2). From the figure, we observe that the secrecy rate without using the RIS increases very slowly with growing $P_t$ and saturates at a lower value. This is because the two Eves are close to Bob such that beamforming at Alice alone can only achieve a very limited secrecy rate. In contrast, the secrecy rates of two schemes with RIS increase significantly with growing $P_t$. The reason is that by configuring the RIS, the signals from the direct channel and the RIS-assisted channel can be combined constructively at Bob, but destructively at Eves, thus enhancing the system security. Moreover, it is observed that the more total transmit power can be allocated within a certain range, the more benefits our proposed SSOC method compared to the reference method will bring. The results also show that the benefits still exist in the case of more than one Eve.

Fig. \ref{Fig.4} shows the performance of the secrecy rate when the number of Eves increases. The total transmit power is 40 dBm. We fix a region with an x-axis of 40 to 45 and a y-axis of 30 to 50, from which we randomly select a certain number of Eves. We perform the random experiment 500 times in the simulation to generate an array of secrecy rates. Then we quantize the array in a way that $95\%$ of the values in it are greater than or equal the result, i.e., $P(R_s(i) \geq \hat{R}_s ) = 95\%$ for $i=1,\dots,500$, where $\hat{R}_s$ denotes the secrecy rate for this region. From the figure, it is observed that the secrecy rates of both schemes for the region are decreasing as the number of Eves increases. This matches with the intuitive expectation that the more Eves are in the same area, the greater the threat to system security. In addition, our proposed SSOC method shows more resilience against Eves in terms of security due to the power allocation compared to the reference method. 

\begin{figure}[tb]
    \centering
    \includegraphics[width=0.5\textwidth]{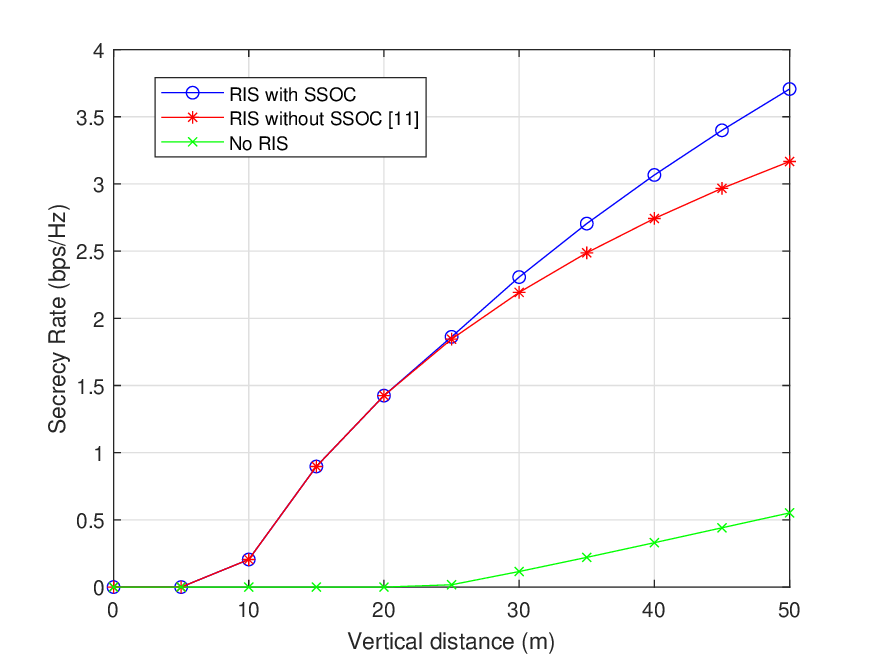}
    \caption{Secrecy rate versus the vertical distance of one Eve.}
    \label{Fig.2}
\end{figure}

\begin{figure}[tb]
    \centering
    \includegraphics[width=0.5\textwidth]{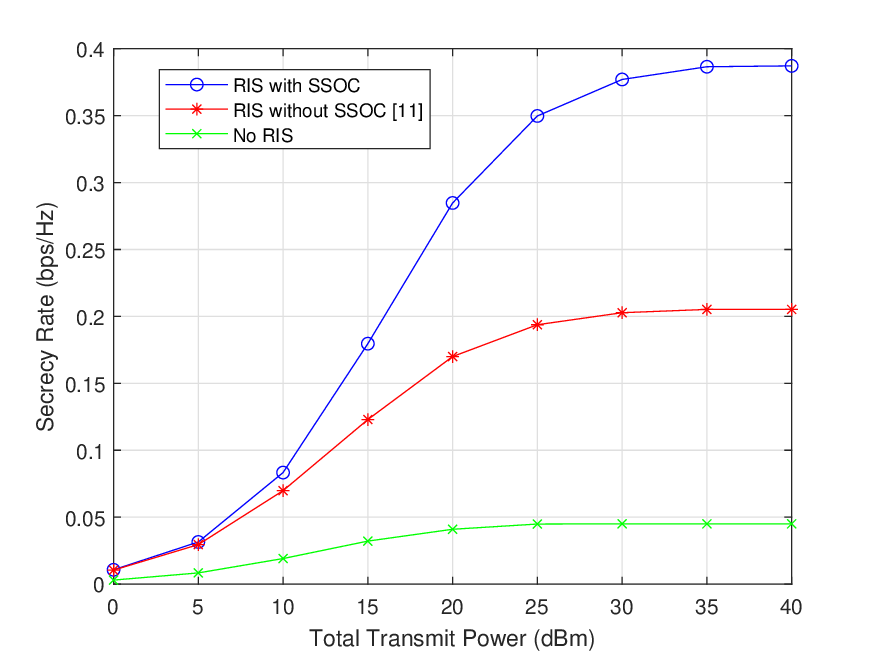}
    \caption{Secrecy rate versus the total transmit power with two Eves.}
    \label{Fig.3}
\end{figure}

\begin{figure}[tb]
    \centering
    \includegraphics[width=0.5\textwidth]{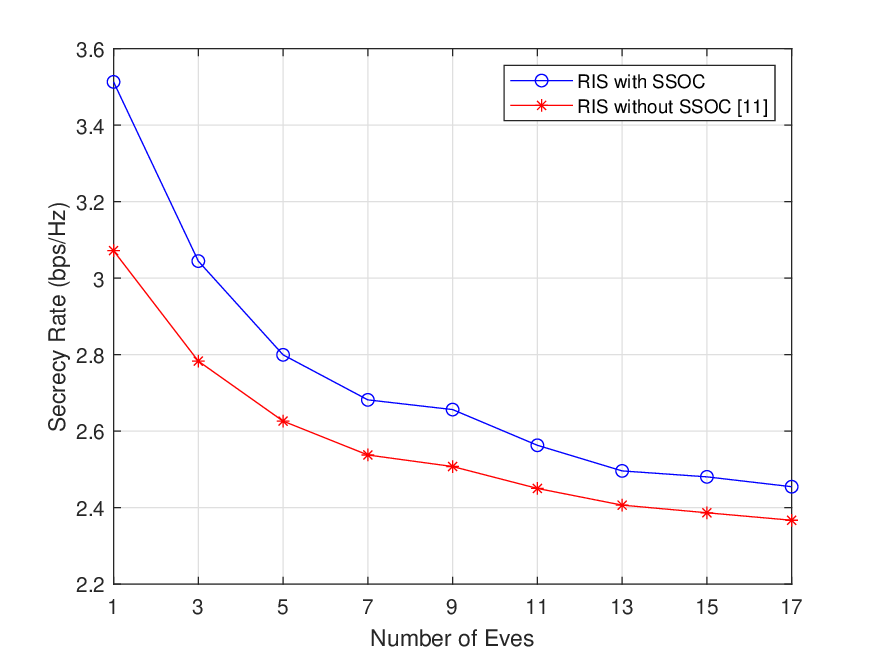}
    \caption{Secrecy rate versus the number of Eves.}
    \label{Fig.4}
\end{figure}

\section{Conclusion} \label{Sec. Conclusion}
In this work, we have proposed a RIS-assisted wiretap channel with SSOC. We have shown that the semantic secure and reliable communication in this regime can be achieved and have derived the achievable secrecy rate. We have designed an optimization algorithm for the achievable secrecy rate which allocates the total available transmit power between the direct link and the RIS-assisted link. The simulation results have demonstrated the security performance enhancement with respect to the total transmit power, the location and the number of eavesdroppers. 
\section{Appendix}
\subsection{Auxiliary Results}\label{Sec. Auxiliary Results}
\begin{lemma}\label{rewrite TV distance}
    The total variation distance can be rephrased as follows:
    \begin{equation}
        ||P_{\boldsymbol{Z}_j^n|\boldsymbol{X}^n(m,\cdot)} - Q_{\boldsymbol{Z}_j}^n||_{\rm{TV}} = \mathbb{E}_{Q_{\boldsymbol{Z}_j}^n}\left[\frac{dP_{\boldsymbol{Z}_j^n|\boldsymbol{X}^n(m,\cdot)}}{dQ_{\boldsymbol{Z}_j}^n}-1\right]^+.
    \end{equation}
\end{lemma}
The proof is given in \cite[Lemma A.6]{bender2021non}.

\begin{theorem}[Tonelli's Theorem \cite{werner2006einfuhrung}]\label{Theo. Tonelli}
    Let $(X,\mathcal{A},\mu)$ and $(Y,\mathcal{B},v)$ be $\sigma$-finite measure spaces and $f: X\times Y \rightarrow [0,+\infty]$ be $\mathcal{A} \otimes \mathcal{B}$-measurable, then
    \begin{align}
        \int_{X \times Y} f d(\mu \otimes v) &= \int_X \left(\int_Y f(x,y) dv(y)\right) d\mu(x) \nonumber \\
        &= \int_Y \left(\int_X f(x,y) d\mu(x)\right) dv(y).
    \end{align}
\end{theorem}

\begin{lemma}\label{Lem. non negative}
    For a non-negative random variable $X$:
    \begin{equation}
        \mathbb{E}[X] = \int_0^{+\infty} P(X \geq t)dt.
    \end{equation}
\end{lemma}
The proof is given in \cite[Equation (21.9)]{billingsley1995probability}.

\begin{lemma}\label{lem. chernoff}
    Let $X_1,\dots,X_n$ be independent random variables, with $0 \leq X_k \leq J$ for each $k$. Let $S_n = \sum_{k=1}^n X_k$ and $\mu = \mathbb{E}[S_n] \leq \mu'$. For any $\epsilon > 0$:
    \begin{equation}
        P(S_n \geq (1+\epsilon)\mu') \leq \exp{\left\{-\frac{\epsilon^2\mu'}{2(1+\epsilon/3)J} \right\}}.
    \end{equation}
\end{lemma}
The lemma above is an extension of Chernoff-Hoeffding inequality \cite{mcdiarmid1998concentration}. The proof of the extension is given in \cite[Appendix 15]{bender2021non}.

\subsection{Proofs of Lemma \ref{Lem. Bound for Typical Terms} and Lemma \ref{Lem. Bound for Atypical Terms}}\label{app:typ-resolv}
\begin{proof}[Proof of Lemma \ref{Lem. Bound for Typical Terms}]
    For $j = 1,\dots,d$ and $t >0$:
    \begin{align}                           
        &\mathbb{E}_{\mathcal{C}_n}\mathbb{E}_{Q_{\boldsymbol{Z}_j}^n}\left[\frac{dP_{j,\mathcal{C}_n,1}}{dQ_{\boldsymbol{Z}_j}^n}-1\right]^+ \nonumber\\                   &=\mathbb{E}_{Q_{\boldsymbol{Z}_j}^n}\mathbb{E}_{\mathcal{C}_n}\left[\frac{dP_{j,\mathcal{C}_n,1}}{dQ_{\boldsymbol{Z}_j}^n}-1\right]^+ \label{eq. Tonelli}\\
         &= \mathbb{E}_{Q_{\boldsymbol{Z}_j}^n}\left[\int_0^{+\infty} P_{\mathcal{C}_n}\left(\left[\frac{dP_{j,\mathcal{C}_n,1}}{dQ_{\boldsymbol{Z}_j}^n}-1\right]^+ \geq t\right)dt\right]\label{eq. non negative}\\
         &= \mathbb{E}_{Q_{\boldsymbol{Z}_j}^n}\left[\int_0^{+\infty} P_{\mathcal{C}_n}\left(\frac{dP_{j,\mathcal{C}_n,1}}{dQ_{\boldsymbol{Z}_j}^n} \geq t+1\right)dt\right] \label{eq. probability term},
    \end{align} where \eqref{eq. Tonelli} follows by applying the Tonelli's Theorem and \eqref{eq. non negative} follows by applying Lemma \ref{Lem. non negative}. The last equality results from the equivalence of the events $\left\{\left[\frac{dP_{j,\mathcal{C}_n,1}}{dQ_{\boldsymbol{Z}_j}^n}-1\right]^+\geq t\right\}(t>0)$ and $\left\{\frac{dP_{j,\mathcal{C}_n,1}}{dQ_{\boldsymbol{Z}_j}^n} \geq t+1\right\}$.

    In the next step, we bound the probability term in \eqref{eq. probability term} by Lemma \ref{lem. chernoff}. Note that the term $\frac{dP_{j,\mathcal{C}_n,1}}{dQ_{\boldsymbol{Z}_j}^n}$ can be regarded as a sum of independent random variables. Thus, bounds for the expectation of the sum and the individual terms of the sum are needed. 

    We begin with the upper bound $\mu'$ for the expectation of the sum $\mu$:
   \begin{equation}
        \mu = \mathbb{E}_{\mathcal{C}_n}[\frac{dP_{j,\mathcal{C}_n,1}}{dQ_{\boldsymbol{Z}_j}^n}] \leq \mathbb{E}_{\mathcal{C}_n}[\frac{dP_{\boldsymbol{Z}_j^n|\boldsymbol{X}^n(m,\cdot)}}{dQ_{\boldsymbol{Z}_j}^n}] =1\coloneq \mu'.
   \end{equation}

   For the bound $J_j$ on individual terms, we apply the definitions of $P_{j,\mathcal{C}_n,1}$ \eqref{eq. P_jC1}, the information density \eqref{eq. def. information density} and the typical set \eqref{eq. typical sets}. Note that $L_1=\exp{\{nR_1\}}$ we have 
   \begin{align*}
       &\frac{1}{{L_1}}\frac{dK^{(\boldsymbol{\theta}){\otimes n}}_{{AE}_j}(\cdot \cap \mathcal{T}_{\epsilon_j}(\boldsymbol{x}^n(m,w)), \boldsymbol{x}^n(m,w))}{dQ_{\boldsymbol{Z}_j}^n}\\
       &= \frac{1}{{L_1}}\exp\log\left(\frac{dK^{(\boldsymbol{\theta}){\otimes n}}_{{AE}_j}(\cdot \cap \mathcal{T}_{\epsilon_j}(\boldsymbol{x}^n(m,w)), \boldsymbol{x}^n(m,w))}{dQ_{\boldsymbol{Z}_j}^n}\right)\\
       &= \exp\{-nR_1\}\exp\log\left(\!\!\frac{dK^{(\boldsymbol{\theta}){\otimes n}}_{{AE}_j}(\cdot \cap \mathcal{T}_{\epsilon_j}(\boldsymbol{x}^n(m,w)), \boldsymbol{x}^n(m,w))}{dQ_{\boldsymbol{Z}_j}^n}\!\!\right)\\
       &= \exp\{-nR_1\}+\exp\{i[\boldsymbol{x}^n;\boldsymbol{z}_j^n]\}\cdot \mathbbm{1}_{(\boldsymbol{x}^n,\boldsymbol{z}_j^n)\in \mathcal{T}_{\epsilon_j}}\\
       & \leq \exp\{-nR_1 + i[\boldsymbol{x}^n;\boldsymbol{z}_j^n]\}\\
       & \leq \exp\{n(I[\boldsymbol{X};\boldsymbol{Z}_j] + \epsilon_j-R_1)\}\coloneq J_j.
   \end{align*}
   
   Substituting $\mu'$ and $J_j$ into Lemma \ref{lem. chernoff} yields
   \begin{align}
       &P_{\mathcal{C}_n}\left(\!\sum_{w=1}^{L_1}\!\!\frac{1}{{L_1}}\frac{dK^{(\boldsymbol{\theta}){\otimes n}}_{{AE}_j}(\cdot \cap \mathcal{T}_{\epsilon_j}(\boldsymbol{x}^n(m,w)), \boldsymbol{x}^n(m,w))}{dQ_{\boldsymbol{Z}_j}^n} \!\geq\! t+1\right) \nonumber\\
   &\quad \leq \exp{\left\{-\frac{t^2}{2(1+t/3)J_j} \right\}}\label{eq. probability bound}.
   \end{align}
   
   Now we use the inequality in \eqref{eq. probability bound} to derive an upper bound on the integral in \eqref{eq. probability term}. Observe that \eqref{eq. probability bound} can be upper bounded by $\exp{\{-\frac{t^2}{3J_j}\}}$ for $0 \leq t <1$ and $\exp{\{-\frac{t}{3J_j}\}}$ for $t>1$. Thus we obtain 
   \begin{align}
       &\int_0^{+\infty} P_{\mathcal{C}_n}\left(\frac{dP_{j,\mathcal{C}_n,1}}{dQ_{\boldsymbol{Z}_j}^n} \geq t+1\right)dt \nonumber\\
       &\leq \int_0^{+\infty}\exp{\left\{-\frac{t^2}{2(1+t/3)J_j} \right\}}dt\nonumber\\
       &\leq \int_0^1 \exp{\{-\frac{t^2}{3J_j}\}}dt + \int_1^{+\infty} \exp{\{-\frac{t}{3J_j}\}}dt\nonumber\\
       &\leq \int_0^{+\infty} \exp{\{-\frac{t^2}{3J_j}\}}dt + \int_0^{+\infty} \exp{\{-\frac{t}{3J_j}\}}dt\nonumber\\
       &= \frac{\sqrt{3\pi}}{2}\sqrt{J_j}+3J_j\nonumber\\
       &= \frac{\sqrt{3\pi}}{2}\exp{\{-\frac{n\epsilon_j}{2}\}}+3\exp{\{-n\epsilon_j\}}\label{eq. randomness rate}\\
       &\leq \exp{\{-\frac{n\epsilon_j}{2}\}}(\frac{\sqrt{3\pi}}{2}+3\exp{\{-\frac{n\epsilon_j}{2}\}})\nonumber\\
       &\leq 5\exp{\{-\frac{n\epsilon_j}{2}\}}\label{eq. randomness R1},
   \end{align} where \eqref{eq. randomness rate} follows by setting the randomness rate for Eve$_j$, $j=1,\dots,d$, to $(I[\boldsymbol{X};\boldsymbol{Z}_i]+2\epsilon_j)$ such that 
   \begin{align*}
       J_j&=\exp\{n(I[\boldsymbol{X};\boldsymbol{Z}_j] + \epsilon_j-R_1)\}\\
       &= \exp\{-n\epsilon_j\}.
   \end{align*} 
   
   Note that we can set one randomness rate $R_1$ for all eavesdroppers:
   \begin{equation}
       R_1 \coloneq  \max_{j=1,\dots,d}(I[\boldsymbol{X};\boldsymbol{Z}_j]+2\epsilon_j),
   \end{equation} in which the inequality \eqref{eq. randomness R1} still holds.
   Finally we have
   \begin{align*}
       \mathbb{E}_{Q_{\boldsymbol{Z}_j}^n}\left[\int_0^{+\infty} P_{\mathcal{C}_n}\left(\frac{dP_{j,\mathcal{C}_n,1}}{dQ_{\boldsymbol{Z}_j}^n} \geq t+1\right)dt\right] &\leq  5\exp{\{-\frac{n\epsilon_j}{2}\}}\\
       &\leq \exp\{-n\beta_{j1}\}
   \end{align*} by specifying $\epsilon_j >0$ and choosing $\beta_{j1}<\frac{\epsilon_j}{2}$.
\end{proof}
\begin{proof}[Proof of Lemma \ref{Lem. Bound for Atypical Terms}]
   We can adopt the idea of proving  channel resolvability in \cite{frey2018resolvability}. For some $\alpha > 1$ and $j=1,\dots,d$:\\
   \begin{align}
        &Q_{\boldsymbol{X},\boldsymbol{Z}_j}^n(\mathcal{T}_{\epsilon_j}^c)\nonumber \\
        &= Q_{\boldsymbol{X},\boldsymbol{Z}_j}^n (\{(\boldsymbol{x}^n,\boldsymbol{z}_j^n):\frac{1}{n}i[\boldsymbol{x}^n;\boldsymbol{z}_j^n] >I[\boldsymbol{X};\boldsymbol{Z}_j] + \epsilon_j \})\nonumber \\
        &= Q_{\boldsymbol{X},\boldsymbol{Z}_j}^n(\{(\boldsymbol{x}^n,\boldsymbol{z}_j^n): \exp{\{(\alpha-1)i[\boldsymbol{x}^n;\boldsymbol{z}_j^n]\}} \nonumber\\
        &\qquad\qquad >\exp{\{(\alpha-1)n(I[\boldsymbol{X};\boldsymbol{Z}_j] + \epsilon_j)\}}\})\nonumber\\
        &\leq \mathbb{E}_{{\boldsymbol{X}^n,\boldsymbol{Z}_j^n}} [\exp{\{(\alpha-1)i[\boldsymbol{x}^n;\boldsymbol{z}_j^n]\}}] \nonumber\\
        &\qquad\qquad\cdot \exp{\{-(\alpha-1)n(I[\boldsymbol{X};\boldsymbol{Z}_j] + \epsilon_j)\}} \label{eq. Markov inequality}\\
        &= \int_{\mathcal{X}^{2\times n}\times{\mathcal{Z}}^{2\times n}} \exp{\{(\alpha-1)i[\boldsymbol{x}^n;\boldsymbol{z}_j^n]\}}Q_{\boldsymbol{X},\boldsymbol{Z}_j}^n(d(\boldsymbol{x}^n,\boldsymbol{z}_j^n))\nonumber\\
        &\quad \cdot \exp{\{-(\alpha-1)n(I[\boldsymbol{X};\boldsymbol{Z}_j] + \epsilon_j)\}}\nonumber\\
        &= \exp\log\{\int_{\mathcal{X}^{2\times n}\times{\mathcal{Z}}^{2\times n}} \left(\frac{dK^{(\boldsymbol{\theta}){\otimes n}}_{AE_j}(\cdot,\boldsymbol{x}^n)}{dQ_{\boldsymbol{Z}_j}^n}(\boldsymbol{z}_j^n)\right)^{\alpha-1} \label{eq. information density} \\
        &\quad  \cdot Q_{\boldsymbol{X},\boldsymbol{Z}_j}^n(d(\boldsymbol{x}^n,\boldsymbol{z}_j^n))\} \cdot \exp{\{-(\alpha-1)n(I[\boldsymbol{X};\boldsymbol{Z}_j] + \epsilon_j)\}}\nonumber
                    \end{align}
    \begin{align}
        &= \exp(-n(\alpha-1)(I[\boldsymbol{X};\boldsymbol{Z}_j]+\epsilon_j-\frac{1}{n}D_{\alpha}[Q_{\boldsymbol{X},\boldsymbol{Z}_j}^n||Q_{\boldsymbol{X}}^nQ_{\boldsymbol{Z}_j}^n])\label{eq. renyi additivity1}\\
        &= \exp(-n(\alpha-1)(I[\boldsymbol{X};\boldsymbol{Z}_j]+\epsilon_j-D_{\alpha}[Q_{\boldsymbol{X},\boldsymbol{Z}_j}||Q_{\boldsymbol{X}}Q_{\boldsymbol{Z}_j}])\label{eq. renyi additivity2}\\
        &\leq \exp{\{-n\beta_{j2}\}} \nonumber
   \end{align} for some 
   \begin{equation}
       0 < \beta_{j2} \leq (\alpha-1)(I[\boldsymbol{X};\boldsymbol{Z}_j]+\epsilon_j-D_{\alpha}[Q_{\boldsymbol{X},\boldsymbol{Z}_j}||Q_{\boldsymbol{X}}Q_{\boldsymbol{Z}_j}])\nonumber,
   \end{equation}
    where \eqref{eq. Markov inequality} is an application of Markov's inequality, and \eqref{eq. information density} follows by the definition of information density \eqref{eq. def. information density}. Due to the additivity of the R\'enyi Divergence \cite[Theorem 28]{van2014renyi}, we get \eqref{eq. renyi additivity2} from \eqref{eq. renyi additivity1}. 
    
Note that under the assumption in Lemma \ref{Bound for Expectation of TV Distance}, $\mathbb{E}[\exp\{t\cdot i[\boldsymbol{X};\boldsymbol{Z}_j]\}]$ exists and is finite for some $t >0$, yielding a finite $D_{\alpha'}[Q_{\boldsymbol{X},\boldsymbol{Z}_j}||Q_{\boldsymbol{X}}Q_{\boldsymbol{Z}_j}]$ for some $\alpha'>1$. This is because for some $t >0$ we have
\begin{align}
       &\mathbb{E}[\exp\{t\cdot i[\boldsymbol{X};\boldsymbol{Z}_j]\}] \\
       = &\mathbb{E}[\exp\{t\cdot \log(\frac{dQ_{\boldsymbol{X},\boldsymbol{Z}_j}}{dQ_{\boldsymbol{X}}Q_{\boldsymbol{Z}_j}})\}]\label{eq. i to ln}\\
        = &\mathbb{E}[(\frac{dQ_{\boldsymbol{X},\boldsymbol{Z}_j}}{dQ_{\boldsymbol{X}}Q_{\boldsymbol{Z}_j}})^t]\\
        = &\int (\frac{dQ_{\boldsymbol{X},\boldsymbol{Z}_j}}{dQ_{\boldsymbol{X}}Q_{\boldsymbol{Z}_j}})^t \cdot Q_{\boldsymbol{X},\boldsymbol{Z}_j}(d(\boldsymbol{x},\boldsymbol{z}_j))\}\\
        = & \int (\frac{dQ_{\boldsymbol{X},\boldsymbol{Z}_j}}{dQ_{\boldsymbol{X}}Q_{\boldsymbol{Z}_j}})^{\alpha'-1} \cdot Q_{\boldsymbol{X},\boldsymbol{Z}_j}(d(\boldsymbol{x},\boldsymbol{z}_j))\}\label{eq. t to alpha}\\
        =& \exp\{(\alpha'-1)D_{\alpha'}[Q_{\boldsymbol{X},\boldsymbol{Z}_j}||Q_{\boldsymbol{X}}Q_{\boldsymbol{Z}_j}]\}\label{eq. int to D},
\end{align} where \eqref{eq. i to ln} follows by the definition of information density \eqref{eq. def. information density}, \eqref{eq. t to alpha} results from the substitution $t = \alpha'-1$ for some $\alpha' >1$ and the definition of R\'enyi divergence \eqref{eq. def. Renyi} yields \eqref{eq. int to D}. 

Observe that the exponential term in \eqref{eq. int to D} is finite under the assumption, leading to a finite exponent $D_{\alpha'}[Q_{\boldsymbol{X},\boldsymbol{Z}_j}||Q_{\boldsymbol{X}}Q_{\boldsymbol{Z}_j}]$ for some finite $\alpha'>1$. Furthermore, $D_{\alpha}[Q_{\boldsymbol{X},\boldsymbol{Z}_j}||Q_{\boldsymbol{X}}Q_{\boldsymbol{Z}_j}]$ is continuous and finite in $\alpha$ for $\alpha<\alpha'$ \cite[Theorem 7]{van2014renyi}. Since $D_{\alpha}[Q_{\boldsymbol{X},\boldsymbol{Z}_j}||Q_{\boldsymbol{X}}Q_{\boldsymbol{Z}_j}] \rightarrow I[\boldsymbol{X};\boldsymbol{Z}_j]$ for $\alpha \rightarrow 1$ \cite[Theorem 5]{van2014renyi}\cite[Section 2.3]{thomas2006elements}, we can choose $\alpha>1$, but sufficiently close to 1 so that the upper bound for $\beta_{j2}$ is positive.
\end{proof}

\subsection{Proofs of Lemma \ref{lem. bound for error one} and Lemma \ref{lem. bound for error two}}\label{app:errors}
 \begin{proof}[Proof of Lemma \ref{lem. bound for error one}] For some $\alpha>1$:
        \begin{align}
            &\mathbb E_{\mathcal{C}_n}[P_{\mathcal{E}_1}(m,w)]
            = \mathbb E_{\mathcal{C}_n}[P_{\boldsymbol{Y}^n|\boldsymbol{X}^n}(\boldsymbol{Y}^n \!\in \!\mathcal{T}^{'c}_{\epsilon} |\boldsymbol{X}^n \!= \boldsymbol{x}^n(m,w))]\nonumber\\
            &=P_{\boldsymbol{X}^n,\boldsymbol{Y}^n}((\boldsymbol{X}^n,\boldsymbol{Y}^n)\in \mathcal{T}_{\epsilon}^{'c})\nonumber\\
            &= P_{\boldsymbol{X}^n,\boldsymbol{Y}^n} (\{(\boldsymbol{x}^n,\boldsymbol{y}^n):\frac{1}{n}i[\boldsymbol{x}^n;\boldsymbol{y}^n] >I[\boldsymbol{X};\boldsymbol{Y}] + \epsilon \})\nonumber \\
            &= P_{\boldsymbol{X}^n,\boldsymbol{Y}^n}(\{(\boldsymbol{x}^n,\boldsymbol{y}^n): \exp{\{(\alpha-1)i[\boldsymbol{x}^n;\boldsymbol{y}^n]\}} \nonumber\\
            &\qquad\qquad >\exp{\{(\alpha-1)n(I[\boldsymbol{X};\boldsymbol{Y}] + \epsilon)\}}\})\nonumber\\
            &\leq \mathbb{E}_{{\boldsymbol{X}^n \!,\boldsymbol{Y}^n}} [\exp{\{(\alpha-1)i[\boldsymbol{x}^n;\boldsymbol{y}^n]\}}] \nonumber\\
            &\qquad\qquad\cdot \exp{\{-(\alpha-1)n(I[\boldsymbol{X};\boldsymbol{Y}] + \epsilon)\}}\label{eq. Markov inequality R}
                        \end{align}
    \begin{align}
            &=  \mathbb{E}_{{\boldsymbol{X}^n \!,\boldsymbol{Y}^n}}[(\frac{dK^{(\boldsymbol{\theta}){\otimes n}}_{AB}}{dQ_{\boldsymbol{Y}}^n})^{\alpha-1}]\! \cdot \exp{\{-(\alpha-1)n(I[\boldsymbol{X};\boldsymbol{Y}] + \epsilon)\}}\nonumber\\
            &= \exp(\!-n(\alpha \!-\!1)(I[\boldsymbol{X};\boldsymbol{Y}]\!+\!\epsilon\!-\!D_{\alpha}[Q_{\boldsymbol{X},\boldsymbol{Y}}||Q_{\boldsymbol{X}}Q_{\boldsymbol{Y}}]))\label{eq. renyi property R}\\
            &\leq \exp{\{-n\Gamma_1\}} \nonumber
        \end{align} for some 
         \begin{equation}
       0 < \Gamma_1 \leq (\alpha-1)(I[\boldsymbol{X};\boldsymbol{Y}]+\epsilon-D_{\alpha}[Q_{\boldsymbol{X},\boldsymbol{Y}}||Q_{\boldsymbol{X}}Q_{\boldsymbol{Y}}]),
        \end{equation}
        where \eqref{eq. Markov inequality R} follows by Markov's inequality, and \eqref{eq. renyi property R} follows by the definition \eqref{eq. def. Renyi} and the additivity \cite[Theorem 28]{van2014renyi} of R\'enyi divergence. We can use similar reasoning as in the proof of Lemma \ref{Lem. Bound for Atypical Terms} to argue that the upper bound for $\Gamma_1$ is positive if $\alpha$ is chosen sufficiently close to 1.
    \end{proof}

     \begin{proof}[Proof of Lemma \ref{lem. bound for error two}]
        \begin{align}
            &\mathbb E_{\mathcal{C}_n}[P_{\mathcal{E}_2}(m,w)] 
            = \mathbb E_{\mathcal{C}_n} [P_{\boldsymbol{Y}^n|\boldsymbol{X}^n}(\boldsymbol{Y}^n \in \mathcal{T}^{'}_{\epsilon} |\boldsymbol{X}^n \!= \boldsymbol{x}^n(m',w'),\nonumber\\
            &\qquad \qquad \qquad \qquad(m',w')\neq (m,w))]\nonumber\\
            &\leq \sum_{m'=1}^L \sum_{w'=1}^{L_1} \mathbb E_{\mathcal{C}_n}[P_{\boldsymbol{Y}^n|\boldsymbol{X}^n}(\boldsymbol{Y}^n \in \mathcal{T}^{'}_{\epsilon} |\boldsymbol{X}^n \!= \boldsymbol{x}^n(m',w'))]\nonumber\\
            &= \exp\{n(R+R_1)\} \!\int_{\mathcal{Y}^{2\times n}}\! \int_{\mathcal{X}^{2\times n}} \!\mathbbm{1}_{(\boldsymbol{x}^n,\boldsymbol{y}^n)\in  \mathcal{T}^{'}_{\epsilon}} Q(d\boldsymbol{x}^n)Q(d\boldsymbol{y}^n)\nonumber\\
            &= \exp\{n(R+R_1)\} \!\int_{\mathcal{Y}^{2\times n}\times \mathcal{X}^{2\times n}} \mathbbm{1}_{(\boldsymbol{x}^n,\boldsymbol{y}^n)\in  \mathcal{T}^{'}_{\epsilon}} \nonumber\\
            &\qquad\qquad\qquad
            \exp\{-i[\boldsymbol{x}^n;\boldsymbol{y}^n]\} Q_{\boldsymbol{X}^n,\boldsymbol{Y}^n}(d(\boldsymbol{x}^n,\boldsymbol{y}^n))\label{eq. information density R}\\
            &\leq \exp\{n(R+R_1)\} \!\int_{\mathcal{Y}^{2\times n}\times \mathcal{X}^{2\times n}} \exp\{-n(I[\boldsymbol{X};\boldsymbol{Y}]+\epsilon)\}\nonumber \\
            &\qquad\qquad\qquad\qquad\qquad Q_{\boldsymbol{X}^n,\boldsymbol{Y}^n}(d(\boldsymbol{x}^n,\boldsymbol{y}^n))\label{eq. typical set R}\\
            &= \exp\{-n(I[\boldsymbol{X};\boldsymbol{Y}]+\epsilon-R-R_1)\}\nonumber\\
            &\leq \exp\{-n\Gamma_2\}\nonumber
        \end{align} for some
        \begin{equation}
            0<\Gamma_2<I[\boldsymbol{X};\boldsymbol{Y}]+\epsilon-R-R_1,
        \end{equation} where \eqref{eq. information density R} follows by the definition \eqref{eq. def. information density} and reformulation of information density, and \eqref{eq. typical set R} follows by the definition of $\mathcal{T}^{'}_{\epsilon}$ \eqref{eq. def. typical set reliability}. Note that $R+R_1<I[\boldsymbol{X};\boldsymbol{Y}]$ such that the upper bound for $\Gamma_2$ is positive.
    \end{proof}

%


\ifCLASSOPTIONcaptionsoff
\newpage
\fi
\bibliographystyle{ieeetr}
\bibliography{ref}

\end{document}